\documentclass[acmsmall, authorversion,nonacm]{acmart}
\usepackage{graphicx} 

\usepackage{tikz}
\usetikzlibrary{shapes}
\usetikzlibrary{bayesnet}
\usetikzlibrary{positioning}
\usetikzlibrary{fit}
\usetikzlibrary{backgrounds}
\usetikzlibrary{calc}

\usepackage{amsmath,amssymb,amsthm}

\definecolor{DarkGreen}{rgb}{0.1,0.5,0.1}

\definecolor{Purple}{rgb}{0.4,0.1,0.5}

\usepackage{amsthm}
\theoremstyle{definition}
\newtheorem{theorem}{Theorem}
\newtheorem{definition}[theorem]{Definition}

\title{The Design and Composition of Structural Causal Decision Processes}

\author{Sebastian Benthall}
\authornote{Corresponding author.}
\affiliation{
  \institution{International Computer Science Institute}
  \country{USA}
}
\affiliation{
  \institution{New York University}
  \country{USA}
}

\author{Alan Lujan}
\affiliation{
  \institution{Johns Hopkins University}
  \department{Krieger School of Arts and Sciences}
  \country{USA}
}

\begin{document}

\begin{abstract}
We present two new classes of causal models of decision-making agents.
Our approach is motivated by the needs of modeling the economics of computing systems. These systems are composed of subsystems and can exhibit endogenous limits on cognitive resources and value discounting.
Structural Causal Decision Models (SCDMs) expand on Structural Causal Influence Models.
Like SCIMs, they explicitly represent the causal relationships between model variables and the payoffs of agent decisions.
Additionally, agent decisions can be constrained by their causal antecedents, and SCDMs can have open root variables for which no probability distribution or structural equation is given.
We show that SCDMs have a well-defined and computationally useful property of composability.
Building on SCDMs, we then define a Structural Causal Decision Process (SCDP) as a recurring SCDM with a discount variable.
SCDPs benefit from the useful composition properties of SCDMs.
Moreover, SCDPs are strictly more expressive than POMDPs because they do not assume rational belief formation.
Indeed, an SCDP can endogenously model the memory-formation process, and is thus useful for modeling resource rational agents in dynamic settings.
SCDPs are also capable of modeling variable discounting, a tool used widely in social scientific modeling.
We pose that SCDPs are a useful framework for policy simulation for the digitial economy, mechanism design for information systems, and digital twin modeling of cyberinfrastructure.
\end{abstract}

\maketitle

Computing systems today are composed of many subsystems, which can include technical components as well as natural and artificially intelligent agents.
Interactions between these components and agents consist of flows of data, which is then stored and used intelligently by other actors.
Realistic models of these systems must be decomposable and reflect endogenous limits on cognitive resources. 
This paper presents a formal framework for constructing such models.
We position this as a new method for flexibly representing potentially intelligent agents in an economy\cite{axtell2025agent}.
Our approach is also consistent with structural causal modeling \cite{pearl2009causality} approaches from computer science and statistics.
We will also perform computational complexity analysis to demonstrate the value of composability to the performance of solution algorithms.

We anticipate applications of this framework to include:
\begin{itemize}
    \item \textbf{Policy simulation.} Exploring the consequences of public policy on the digital economy.
    \item \textbf{System design.} Modeling data flows in dynamic systems that include intelligent agents, with parameters that can be controlled for mechanism design.
    \item \textbf{Digital twin modeling.} Realistic modeling that combines strategic actors with physics-informed neural networks (PINN) for scenario-based risk analysis.
\end{itemize}

The framework developed in this paper stands on much prior work.
Most closely, it builds to something similar to a dynamic influence model \cite{neapolitan2004learning}. However, we build on more recent work on structural causal games \cite{hammond2023reasoning} as well as the modeling concerns of computational economics, and more recent advances in deep learning, to arrive at a new design.

In this paper, we will focus on the definition of composable, causal environments for a single agent.
We see this as the first step.
Future work will expand the definition to multi-agent environments,
and address algorithms for learning the equilibrium strategies of the agents and using the models for statistical analysis.

\subsection{Contributions}

Section \ref{sec:prior-work} situates this article in prior literature and recent developments in computational economics and social science, AI safety, composable world models, and variable discounting.

Section \ref{sec:scdm} introduces the Structural Causal Decision Model, a variation on Structural Causal Influence Models (see Definition \ref{def:scim}, \cite{hammond2023reasoning}) that supports constraints on decision variables and root variables with no probability distribution.
We will show how SCDMs can be composed and decomposed, and how, in a special case of sequential decomposition, this readily improves the efficiency of solving for optimal decision rules.
This is illustrated with 2-period consumption saving problems, one of which introduces habit-formation.

Section \ref{sec:scdp} expands on the SCDM to show how it can be converted into a dynamic stochastic optimization problem.
The new construct, a Structural Causal Decision Process (SCDP), can also be decomposed to improve solution efficiency.
This is illustrated with a dynamic model that includes both consumption and portfolio allocation.

Section \ref{sec:scdp-pomdp-rr} discusses the expressive potential of SCDPs, showing that they are more expressive than both MDPs and POMDPs.
We demonstrate that in exceeding the expressiveness of POMDPs, SCDPs are able to model resource rationality by endogenizing cognitive constraints and tradeoffs of agents.

Section \ref{sec:scdp-variable-discount} presents an adjustment to SCDPs that enables them to represent varying discount factors.
We show that an SCDP can express models from macroeconomics where the discount factor, or patience, of consumers varies as a Markov process.

In sum, this article builds up a new general formalism for representing dynamic optimization problems that is expressed in terms of Pearlian causation, has useful properties of compositionality, and can model resource rationality and variable or abstract discounting.
We intend SCDMs and SCDPs to be a new computational class of models that is well-suited to social scientific modeling of complex sociotechnical systems, such as human-AI hybrid ecosystems.
We present the formal framework and motivation for these models in anticipation of ongoing work that expands them to represent multi-agent systems, and develops efficient algorithms for fitting these models to data and solving them for optimal decision rules.

The implementing project is called \texttt{scikit-agent}.
It's documentation\footnote{\url{https://scikit-agent.github.io/scikit-agent/}} and source code\footnote{\url{https://github.com/scikit-agent/scikit-agent}} are available online.

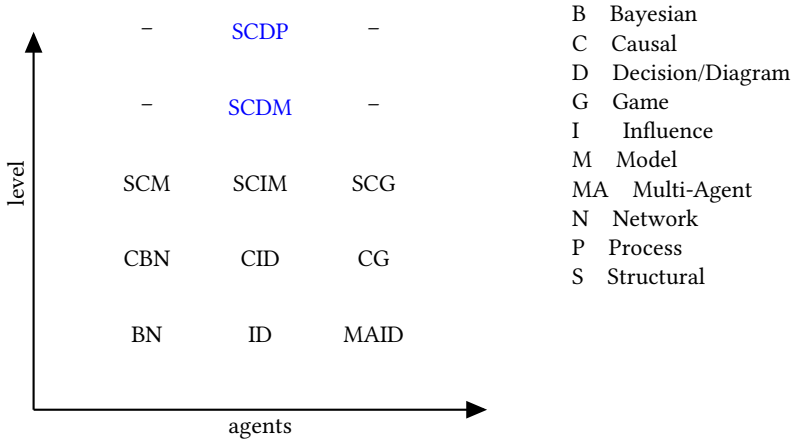
\begin{figure}
\begin{tikzpicture}[
    node distance=1.2cm and 1.5cm,
    every node/.style={font=\small}
]

\draw[thick, ->] (0,0) -- (0,5) ;
\node[above, rotate=90, anchor = south] at (0,3) {level};
\draw[thick, ->] (0,0) -- (6,0);
\node[below] at (3,0) {agents};

\node at (1.5, 5) {--};
\node[color = blue] at (3, 5) {SCDP};
\node at (4.5, 5) {--};

\node at (1.5, 4) {--};
\node[color = blue] at (3, 4) {SCDM};
\node at (4.5, 4) {--};

\node at (1.5, 3) {SCM};
\node at (3, 3) {SCIM};
\node at (4.5, 3) {SCG};

\node at (1.5, 2) {CBN};
\node at (3, 2) {CID};
\node at (4.5, 2) {CG};

\node at (1.5, 1) {BN};
\node at (3, 1) {ID};
\node at (4.5, 1) {MAID};

\node[anchor=west, align=left] at (7, 3.5) {
    B\quad Bayesian\\
    C\quad Causal\\
    D\quad Decision/Diagram\\
    G\quad Game\\
    I\quad\phantom{M}Influence\\
    M\quad Model\\
    MA\quad Multi-Agent\\
    N\quad Network\\
    P\quad Process\\
    S\quad Structural
};

\end{tikzpicture}
\caption{We extend the taxonomy of causal models from \citet{hammond2023reasoning}, and in this paper introduce SCDM and SCDP.
In the original diagram, the causal hierarchy (associational,
interventional, and counterfactual) forms the vertical axis and the number of agents (0, 1, and
$N$) forms the horizontal axis.
We add two additional elements to the causal hierarchy: functional, and dynamic. This paper introduces two new classes of causal models, Structural Causal Decision Models (SCDM) and Structural Causal Decision Processes (SCDP) that illustrate these steps in the hierarchy.}
\end{figure}

\section{Prior work}
\label{sec:prior-work}

This work aims to synthesize recent developments in several areas of social and computer science as a foundation for asking and answering new kinds of questions about computing systems.

\subsection{Economics and computational social science}

Agent-based modeling (ABM) has intrigued social scientists for decades.
Computational tools and ideas from physics informed computational sociology and enabled early work on studying ``artificial societies'' \cite{epstein1996growing}.
ABM has since been understood to be a bridge between disciplines \cite{axelrod2006agent}.
On the other hand, economics has typically viewed ABM with some skepticism, preferring models that are more explicitly decision-theoretic \cite{blume2015agent}.

However, in their recent reformulation of the relationship between ABM and economics, 
\citet{axtell2025agent} paint an inspiring picture of its current and future role within the field.
The availability of ground truth micro-data and cheaper compute has made mathematical assumptions about aggregate or emergent level phenomena, which has been the mainstay of many earlier methods, less appealing than direct simulation of large populations of simulated agents.
These simulations allow the modeler to relax strict and often unrealistic assumptions of ``rational expectations'' and address cases of bounded rationality.
That lets modelers achieve new levels of detail and realism.
Central banks stand out as an example of vital institutions that use fine-grained, multi-sector ABMs to study economic outcomes.
A single ABM can address several different aspects of the economy, including interbank credit markets, financial markets like the stock market, housing markets, and even the effects of climate change \cite{borsos2025agent}.

The use of ABMs for policy simulation is not without challenges.
When, because of model complexity, agents can only approximately optimize their behavior, this can lead to unrobust learning, which can be mitigated by careful training procedures 
\cite{agrawal2025robust}.
We maintain that given a formally well-defined decision-theoretic and structural model, these challenges may be surmountable.



\subsection{Computer science and AI safety}

We draw on techniques from the field of artificial intelligence.
While economics has been slow to take up Pearl's causal graphical modeling approach \citet{pearl1994probabilistic} , this is common in computer science.
Causal modeling has been used extensively 
in modeling the social effects of automated systems in such fields as fairness \cite{creager2020causalmodelingfairnessdynamical, mhasawade2021causalmultilevelfairness} and privacy \cite{benthall2019situated}. 
More recent work has combined this causal modeling of the impact of AI systems with explicit modeling of the intelligence -- or goal-orientedness -- of the AI systems themselves to addres the problem of aligning AI systems with human users.

The first synthesis of Pearlian causation and game theory was done by \citet{koller2003multi} with the introduction of Multi-Agent Influence Diagrams (MAIDs).
This work shows that Bayesian networks can be augmented with decision and utility variables that are assigned to different agents, and that this provides a compact way to represent multi-agent games that would be intractable to model in extensive form.
New formal work has refined and expanded on the MAID framework and better aligned it with structural causal modeling, introducing Structural Causal Influence Models (SCIMs) and Structural Causal Games (SCG) \cite{hammond2021equilibrium, hammond2023reasoning}.
There is a software implementation, \texttt{PyCID}, that encapsulates some of the work of this team of researchers \cite{fox2021pycid}.

These new frameworks have been used to model AI safety issues, such as \emph{reward tampering} where an AI system is able to directly modify its reward function \cite{everitt2021reward}, proposing AI system designs that would prevent the use of proxies for protected categories in classification tasks, and prevent a recommendation system from polarizing the politics of its users \cite{everitt2021agent}, as well as operationalizing what it means for an AI system to be honest, deceptive \cite{kenton2023discovering}, or manipulative \cite{carroll2023characterizing}.

Beyond its usefulness for modeling AI and sociotechnical systems, there is a strong case to be made that causal representations of the world are necessary for robustness to distributional shifts \cite{richens2024robust}.
In general, causal learning and modeling is considered by many to be the gold standard of machine intelligence, and we intend our framework to adhere to this standard.

\subsection{Composability}
\label{sec:composability}

A further motive for our approach is the desire for \emph{composability} in modeling and simulation.
By this we mean that it is useful if models can be decomposed into components that can be analyzed independently, and if they can be composed into new models.
This is motivated by concerns that arise in both optimization, and in simulation.

\paragraph{In optimization}
Composability can improve the tractability of an optimization problem embedded in the decision-theoretic model.
One of the attractive properties of the Multi-Agent Influence Diagram \cite{koller2003multi} paradigm is that it allows a model to be decomposed into subgames which can be more efficiently solved to discover Nash equilibrium strategy profiles.
In computational economics, \citet{lujan2026egmn} shows how a complex consumption-saving problem can be decomposed into a multi-stage problem, which improves the computational efficiency of solving it.

Several lines of prior work have explored how to use the decomposability of the transition function of a Markov Decision Process (MDP) to improve training performance. These approaches include factoring the subgraphs of a dynamic Bayesian network \cite{boutilier2000stochastic}, clustering the state space based on the transition function \cite{mannor2004dynamic}, and decomposing the problem into a hierarchy of MDP subtasks \cite{dietterich2000hierarchical}.

\paragraph{In simulation} Beyond that concern, composability is a useful property when using models and simulation to understand and design complex systems.
Complex systems are defined as those that consist of multiple distinct parts which interact to produce emergent properties that cannot be easily reduced to individual behavior.
So modeling these systems generally must proceed from the bottom up.
Crucially, the microstructure of these components can be empirically validated separately from the emergent outcomes.
Thus, compositional systems for modeling and simulation have a long 
history \cite{kasputis2000composable, august2005achieving, balci2011achieving},
especially for the purpose of system design \cite{paredis2001composable}, digital twins \cite{van2023achieving}, and
studying complex systems \cite{zhu2019reusability, wagner2024comparing}.
Recently, AI researchers have turned their attention to composable world models \cite{sehgal2023neurosymbolic} and environments for LLMs \cite{zhang2024combo, wang2025modeling} and robot testing \cite{ortega2024composable}.
Researchers have grounded composability of simulations in terms of causality \cite{wang2022respecting} and interconnection structure \cite{neary2023compositional} when simulating physical systems, suggesting the viability of a Pearlian approach that is both causal and graphical.

\subsection{Resource rationality}
\label{sec:resource-rationality}

One recent development in the modeling of cognition -- whether human or artificial -- is the emergence of \emph{resource rational analysis} \cite{lieder2020resource}, a reformulation of earlier versions of \emph{rational analysis} in cognitive science \cite{chater1999ten, anderson1991human}.
Rational analysis attempts to provide teleological explanations of cognitive operations as being a rational action given the incentive structure of the cognitive system.
It interprets cognitive behavior as being directed towards goals even when it appears at first to be suboptimal.
It is used as a way of narrowing down potential hypotheses to those which have evolutionary plausibility.
One of the considerations in rational analysis is the availability of cognitive resources to solve problems; limited resources are one reason why people or animals may employ heuristics to solve a complex cognitive problem.

Resource rational analysis restates the core theses of rational analysis while repositioning them on a decision-theoretic framework in which the cost of cognitive resources are explicitly represented in the utility function of the agent.
Agents choose to optimize their decision rule or policy taking into account both expected results and cognitive costs. This is proposed as a useful tool for modeling realistic agents.
The concept of resource rationality has been taken up by recent proposals about AI safety and alignment, because of the challenges of aligning an AI system with humans who have limited cognitive abilities and perhaps not even rational preference structures \cite{zhi2025beyond, levine2025resource}.
One motive for the modeling framework presented in this article is that it provides a way to model dynamic optimization problems in which cognitive resources are endogenously chosen and costly.

\subsection{Variable discounting}
\label{sec:lit-review-variable-discounting}

In dynamic programming and reinforcement learning settings, there is normally a constant discount factor $\gamma $ or $\beta \in (0,1)$.
However, in social scientific modeling, there are many uses of variable discount factors that may be subject to exogenous or even endogenous processes.
We briefly discuss this work to motivate the inclusion of variable discounting in the SCDP framework we introduce.

\paragraph{Time inconsistency and hyperbolic discounting} The seminal work of \citet{strotz1956myopia} challenged the orthodoxy of exponential discounting by showing that time-variant discount rates imply dynamically inconsistent preferences: plans that are optimal today may no longer appear optimal tomorrow, even absent new information. This insight launched a large literature in behavioral economics.

\citet{laibson1997golden} introduced quasi-hyperbolic (or $\beta$-$\delta$) discounting, which approximates hyperbolic discounting in discrete time while retaining analytical tractability. Under this specification, agents discount the immediate future more heavily than distant periods, capturing the ``present bias'' observed in experimental settings. A key result is that sophisticated agents with quasi-hyperbolic preferences undersave relative to the exponential benchmark. \citet{harris2001dynamic} extended this analysis, proving existence of equilibrium and deriving a generalized Euler equation for hyperbolic consumers.

\paragraph{State-dependent discounting} More recently, \citet{stachurski2021dynamic} generalized discrete-time infinite-horizon dynamic programming to allow the discount factor to depend on the state. Rather than requiring $\beta < 1$ uniformly, they impose that the discount factor process be strictly less than one \emph{on average in the long run}. Under this condition, the standard optimality results are recovered: Bellman's principle of optimality holds, and both value function iteration and policy function iteration converge. Their framework accommodates recursive preferences and unbounded rewards, with applications to asset pricing and interest rate dynamics.

\paragraph{Quantitative macroeconomics} Variable discounting plays an important role in heterogeneous-agent macroeconomics. \citet{krusell1998income} introduced discount factor heterogeneity to match the extreme concentration of wealth observed in the data: agents with higher $\beta$ accumulate more assets in the long run. \citet{krusell2003consumption} embedded quasi-geometric discounting in the neoclassical growth model and showed that the resulting intrapersonal game admits a continuum of equilibria, complicating both theory and computation. \citet{krusell2002equilibrium} analyzed welfare and policy implications in this setting. More recently, \citet{cao2020recursive} proved existence of recursive equilibrium in the Krusell-Smith economy with state-dependent discount factors, extending the theoretical foundations for this class of models.

\vspace{1em}

Having situated and motivated our modeling approach, we will now proceed to present its core features.
We will draw examples from computational economics, employ causal modeling tools from computer science, and demonstrate how our system enables composability.

\section{Structural Causal Decision Models (SCDM)}
\label{sec:scdm}

We have motivated this work with reference to a broad scope of social scientific modeling.
While this work is in service to that agenda of modeling multi-agent systems, in this paper we will focus on the special case of a single agent.

\subsection{Example: Two-period consumption model}

Before presenting the formal definition of a Structural Causal Decision Model, we illustrate the core concepts with a canonical example from economics: the Fisher two-period consumption problem \citep{fisher1930theory}.

Consider a consumer who lives for two periods and seeks to maximize lifetime utility:
$$\max_{c_1, c_2} u(c_1) + \beta u(c_2)$$
where $c_t$ denotes consumption in period $t$, $u(\cdot)$ is a strictly increasing, strictly concave utility function, and $\beta \in (0,1)$ is a time preference factor.

The consumer begins period 1 with resources $b_1$ (for ``bank balances'')\footnote{We use $b_t$ here to align with standard notation in the consumption literature; later sections use $w_t$ for wealth when discussing more general dynamic models.} and income $y_1$. Resources not consumed become end-of-period assets $a_1$:
$$a_1 = b_1 + y_1 - c_1.$$
These assets earn a gross return $R = 1 + r$, yielding period 2 resources $b_2 = R \cdot a_1$. In the final period, the consumer receives income $y_2$. They then choose $c_2 < b_2 + y_2$ to maximize utility.
The optimal policy at this stage is, quite trivially, to consume all available resources $c_2 = b_2 + y_2$.

At first glance, this problem appears to involve three interrelated choices: how much to consume today ($c_1$), how much to save ($a_1$), and how much to consume tomorrow ($c_2$). However, inspection of the structural equations shows that these apparent decisions collapse into a single degree of freedom. Once the consumer chooses $c_1$, the asset equation determines $a_1$, which in turn determines $b_2$.
The decision rule for $c_2$ is very simple to derive and compute.

This observation motivates the formal framework we develop below. We distinguish between \emph{state variables} that describe the system at a point in time ($b_1$, $b_2$), \emph{decision variables} under the agent's control ($c_1, c_2$), and \emph{utility variables} that enter the objective ($u(c_1)$, $u(c_2)$). These are linked by \emph{structural equations} (deterministic relationships such as $a_1 = b_1 + y_1 - c_1$) and subject to \emph{constraints} on feasible decisions ($0 \leq c_1 \leq b_1 + y_1; 0 \leq c_2 \leq b_2 + y_2$).

The structure also exhibits natural decomposition, illustrated in Figure~\ref{fig:two-period-id}. The variable $a_1$ serves as a ``bridge'' connecting two stages: period 1, where the agent chooses $c_1$ given $b_1$, and period 2, where $c_2$ is determined given $b_2$. This decomposition has computational significance: because $c_2$ is fully determined by the bridge variable, the second-stage ``problem'' can be solved first, yielding a continuation value function $v_2(b_2)$; this value function then informs the first-stage optimization. The sequential structure reduces a joint optimization over $(c_1, c_2)$ to a sequence of simpler problems, foreshadowing the decomposition techniques we develop in later sections.

\begin{figure}[h]
\centering
\begin{tikzpicture}[
    state/.style={circle, draw, minimum size=0.9cm, font=\small},
    decision/.style={rectangle, draw, minimum size=0.9cm, font=\small},
    utility/.style={diamond, draw, minimum size=0.7cm, font=\small, inner sep=1pt},
    noise/.style={circle, draw, double, minimum size=0.9cm, font=\small},
    discount/.style={regular polygon, regular polygon sides=5, draw, minimum size=1.1cm, font=\small, inner sep=1pt},
    >=stealth
]

\node[state] (b1) at (0, 0) {$b_1$};
\node[state] (a1) at (5, 0) {$a_1$};
\node[state] (R) at (6.5, -1) {$R$};
\node[state] (b2) at (8, 0) {$b_2$};

\node[state] (y1) at (0, -2) {$y_1$};
\node[state] (y2) at (8, -2) {$y_2$};

\node[decision] (c1) at (1.5, -4) {$c_1$};
\node[utility] (u1) at (3, -4) {$u_1$};
\node[decision] (c2) at (9.5, -4) {$c_2$};
\node[utility] (u2) at (11, -4) {$u_2$};

\node[discount] (beta) at (10, -2) {$\beta$};

\draw[->] (b1) -- (c1);
\draw[->] (y1) -- (c1);

\draw[->] (c1) -- (u1);

\draw[->] (b1) -- (a1);
\draw[->] (y1) -- (a1);
\draw[->] (c1) -- (a1);

\draw[->] (a1) -- (b2);
\draw[->] (R) -- (b2);

\draw[->] (b2) -- (c2);
\draw[->] (y2) -- (c2);

\draw[->] (c2) -- (u2);

\draw[->] (beta) -- (u2);

\begin{scope}[on background layer]
    \node[draw=gray, dashed, rounded corners, fit=(b1)(y1)(c1)(u1)(a1), inner sep=0.4cm] (box1) {};
    \node[draw=gray, dashed, rounded corners, fit=(a1)(R)(b2)(y2)(c2)(u2)(beta), inner sep=0.4cm] (box2) {};
\end{scope}

\node[below=0.3cm of box1.south, font=\footnotesize\itshape] {Period 1};
\node[below=0.3cm of box2.south, font=\footnotesize\itshape] {Period 2};

\end{tikzpicture}
\caption{Influence diagram for the two-period consumption problem. Circles denote state variables, rectangles denote decision variables, diamonds denote utility, pentagons denote discount factors. The bridge variable $a_1$ (end-of-period savings) connects the two periods; the boxes overlap on this node to indicate it belongs to both components. The edge from $\beta$ to $u_2$ indicates discounting of future utility.}
\label{fig:two-period-id}
\end{figure}
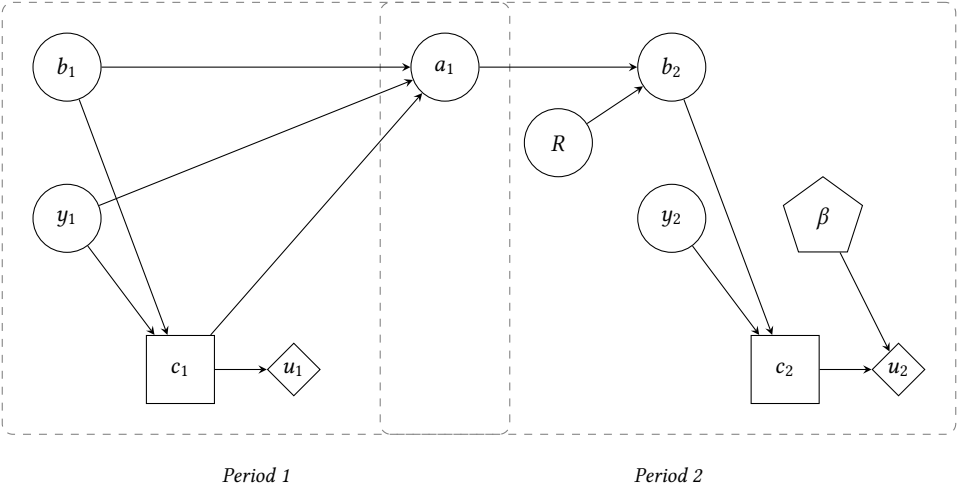

\subsection{Structural Causal Decision Models}

This section introduces a mathematical construct, which we will call a Structural Causal Decision Model (SCDM).
This is closely akin to a different structure, the Structural Causal Influence Model of \citet{hammond2023reasoning, koller2003multi}, with a few key differences. (SCIMs are defined in Definition \ref{def:scim} in the Appendix.)
Like the SCIM, it builds on both the prior work on influence models \cite{shachter1986evaluating}, which represent decision-theoretic problems with a graph, and Pearlian graphical causal modeling \cite{pearl1994probabilistic}.
And SCDM augments an SCIM by addign explicit constraints on the optimization problem, as well as open 'inputs' to the model, which we will call 'root variables'.

We will use $\mathbf{Pa}_V$ to denote the parents of a variable $V$, given a graph; $\mathbf{pa}_V$ is a potential value of the parents of that variable.
"$\text{dom}$" refers to the function that returns the domain (the set of all possible values) of a variable.

\begin{definition}[Structural Causal Decision Model (SCDM)]
A \emph{structural equation influence model} $(\mathbf{V}, \mathbf{Z}, \mathcal{E}, (\mathbf{X}, \mathbf{D}, \mathbf{U}), \mathbf{Pr}, \mathbf{f}, \boldsymbol{\theta})$ consists of:
\begin{itemize}
\item A set of endogenous variables $\mathbf{V}$
\item A set of exogenous variables $\mathbf{Z}$.
\item An edge set $\mathcal{E}$ such that the graph $\mathcal{G} = (\mathbf{V} \cup \mathbf{Z}, \mathcal{E})$ is a DAG over $\mathbf{V} \cup \mathbf{Z}$, in which $\forall Z \in \mathbf{Z}, \mathbf{Pa}_Z = \varnothing$.
\item $\mathbf{V}$ is partitioned into: \begin{itemize}
\item $\mathbf{X}$, state variables
\item $\mathbf{D}$, decision variables\footnote{The influence diagram literature rarely intersects with the control theory literature. Decision variables and control variables are roughly synonymous.}
\item $\mathbf{U}$, utility variables
\end{itemize}
\item A parameterized probability distribution $\mathbf{Pr}(\boldsymbol{\theta_E})$ over all exogenous variables $\mathbf{Z}$.
\item $\mathbf{f} = \{f_V: \text{dom}(\mathbf{Pa}_V) \times \text{dom}(\boldsymbol{\theta}) \rightarrow \text{dom}(V) | V \in \mathbf{V} \setminus \mathbf{D} \text{ and } \mathbf{Pa}_V \neq \emptyset \}$  are structural equations governing all variables in $\mathbf{V}$ that are not decision nodes or parentless (root) nodes.
\item $\mathbf{\Gamma} = \bigcup_{D \in \mathbf{D}} \{\Gamma_D\}$ where $\Gamma_D : \text{dom}(\mathbf{Pa}_D) \times \text{dom}(\boldsymbol{\theta}) \rightarrow \mathbb{P}(\text{dom}(D))$ are constraints on the actions allowable at the decision variables.
\item Parameters $\boldsymbol{\theta}$
\end{itemize}
\label{def:scdm}
\end{definition}

SCDM is a departure from SCIM (Definition \ref{def:scim}, \citet{hammond2023reasoning}) in just a few ways:

\paragraph{Constraints} We have introduced the decision constraints $\boldsymbol{\Gamma}$, which are not included in the original definition. Introducing these constraints makes it easier to synthesize the SCIM and SCG constructs with the binding constraints of many control theory problems. An example constraint is the budget constraint ($c_1 \leq b_1$) in the consumption saving problem.

\paragraph{Structural functions} Rather than expressing the relationships between variables as ``deterministic conditional probability distributions'', we define deterministic functions $\mathbf{f}$ for each endogenous non-decision variable. Conversion to a conditional probability distribution is straightforward, and so this is largely a matter of style. However, unlike an SCIM, it is possible to have a state variable over which there is no \emph{unconditional} probability distribution or governing structural equation. These are the variables $V$ without parents.

\paragraph{Limiting noise} In a classic Pearlian structural causal model, each deterministic variable has a noise term $\epsilon_V$. While SCDMs can express such a thing, as a matter of style we indicate each exogenous noise variable separately.
\paragraph{Continuous valued models.} While both SCIMs and SCDMs are general with respect to whether variables are continuous or discrete, we have designed SCDMs to be well-suited to models with continuous state, shock, and decision values. The structural functions $f \in \mathbf{f}$ may be differentiable or invertible, which opens up more efficient solution methods \cite{lujan2026egmn}
\vspace{1em}

We will not ground this form of modeling in measure theory in this paper.
We will assume that utility variables $\mathbf{U}$ range over the real numbers for the purposes of presentation, though the structure is more general than this.

\subsubsection{Roots and leaves}

We allow an SCDM to have state variables $X$ with no parents, $\mathbf{Pa}_X = \varnothing$ (this is a departure from SCIM models). There are also variables with no descendants.

\begin{definition}[Roots and leaves]
    Given an SCDM $\mathcal{M} = (\mathbf{V}, \mathbf{Z}, \mathcal{E}, (\mathbf{X}, \mathbf{D}, \mathbf{U}), \mathbf{Pr}, \mathbf{f}, \boldsymbol{\theta})$, let the \emph{roots} of the SCDM be  $\tilde{\mathbf{X}} = \{X | X \in \mathbf{X} \text{ and } \mathbf{Pa}_X  = \varnothing \}$, i.e., the chance variables with no parents. Let the \emph{leaves} of the SCDM be the variables $V \in \mathbf{V}$ with no descendants.
\end{definition}

Note that while exogenous shock variables in $\mathbf{Z}$ have no parents, they are not considered to be root variables in this sense of the term.
For an SCDM, root variables will have no governing structural equation or probability distribution.

The SCDM can be seen as a stochastic function from root values $\mathbf{\tilde{x}} \in \text{dom}(\mathbf{\tilde{X}})$ to the values of all other variables, including the utility variables.
The root nodes begin the flow of causation through the model.
At the decision variables, an agent chooses a rule that governs the flow of values.

\subsubsection{Decisions}

SCDMs are in the family of \emph{influence models} because they define a decision-theoretic problem for the agent.

\begin{definition}[Decision rule]
    Given an SCDM $(\mathbf{V}, \mathbf{Z}, \mathcal{E}, (\mathbf{X}, \mathbf{D}, \mathbf{U}), \mathbf{Pr}, \mathbf{f}, \boldsymbol{\theta})$, a decision rule $\pi_D$ for a decision variable $D \in \mathbf{D}$ is a function $\pi_D(\mathbf{Pa_D})$ of the form $\pi_{D}: \text{dom}(\mathbf{Pa}_D)  \rightarrow \text{dom}(D)$, such that $\forall \mathbf{p} \in \text{dom}(\mathbf{Pa}_D), \pi_D( \mathbf{p}) \in \Gamma_D( \mathbf{p})$.\footnote{In other words, the decisions according to the rule must fall within the decision constraints.} A set of decision rules $\mathbf{\pi}$ for all decision variables $\mathbf{D}$ is called a \emph{policy profile}.\footnote{Mixed strategy decision rules can be introduced into SCDMs by including an exogenous noise variable as a parent of the decision nodes.}
\end{definition}

A policy profile and values for the root variables provide all that is needed to \emph{induce} an SCDM into a fully specified Structural Causal Model (see Definition \ref{def:scm}), which characterizes the joint distribution over all model variables.
This is done by turning each decision variable $D$ into a state variable governed by the corresponding decision rule $\pi_D$, and assigning values to the root variables.

For any variable $V \in \mathbf{V}$, we can construct the deterministic function $f^*_V(\tilde{\mathbf{x}}, \mathbf{z}, \mathbf{\pi})$ for the value of that variable given the realization of root variables, shocks, and decision rules.

Under usual rationality assumptions, the agent's rational behavior is to choose the decision rules that maximize expected utility.
We can define the \emph{root value function} of the SCDM to be the maximum possible expected utility, given root values.

\begin{definition}[Root value function of an SCDM]
We define the \emph{root value function} of an SCDM $\mathcal{M} = (\mathbf{V}, \mathbf{Z}, \mathcal{E}, (\mathbf{X}, \mathbf{D}, \mathbf{U}), \mathbf{Pr}, \mathbf{f}, \boldsymbol{\theta})$ to be:

   \begin{equation}
        v_\mathcal{M}(\tilde{\mathbf{x}}) = \max_{\mathbf{\pi}} E_\mathbf{Z} \left[ \sum_{U \in \mathbf{U}} f^*_U( \mathbf{\tilde{X}} = \tilde{\mathbf{x}}, \mathbf{z}, \boldsymbol{\pi}) \right]
        \label{eq:scim-root-value}
    \end{equation}
    \label{def:scim-root-value}
\end{definition}

The function $v_\mathcal{M}: \text{dom}(\tilde{\mathbf{X}}) \rightarrow \mathbb{R}$ is over the values of the root variables of $\mathcal{M}$ and reflects the optimal expected value of an agent subject conditioned on the given values of the root variables.
Note that the choice of policy profile $\pi$ occurs, logically, before the realization of the shocks.

Note that this is an optimization over decision rules, not an optimization over ``actions'', or values of decision variables.
This is a significant distinction, because the space of possible decision rules $\boldsymbol{\pi}$ is not the same as the space of possible functions from root values to decisions ($\text{dom}(\mathbf{\tilde{X}}) \rightarrow \text{dom}(\mathbf{D})$).
We draw a distinction between the optimal decision rules -- which are subject to the limited information available to the agent -- and the optimal actions which would be taken by an omniscient agent under the circumstances of the root values:

   \begin{equation}
        \mathbf{d}^*_\mathcal{M}(\tilde{\mathbf{x}}) = \arg \max_{\mathbf{d} \in \text{dom}(\mathbf{D})} E_\mathbf{Z} \left[ \sum_{U \in \mathbf{U}} f^*_U(\mathbf{\tilde{X}} = \tilde{\mathbf{x}}, \mathbf{z}, \mathbf{D} = \mathbf{d}) \right]
        \label{eq:scim-optimal-action}
    \end{equation}

SCDMs precisely define the information available at each decision node and for a given model, it is possible that a policy profile that guarantees optimal decision values does not exist.

\subsubsection{Composition}

We have discussed the motivations for model composition in Section \ref{sec:composability}.
Composability enables the reuse and independent empirical validation of model components.
We also anticipate ways in which composability can be leveraged to improve learning algorithms.
Sometimes composability can be leveraged for computational performance.

In this paper, we will focus on a special case of composition: cases where an SCDM can be sequentially decomposed, meaning that it can be decomposed into two components, such that one follows 'after' the other.
This has the computational benefit that the problem of selecting the optimal policy profile can be divided into nested subgames such that the last subgame can be solved first, and so on, through backwards induction.

The overall idea of composing SCDMs is quite simple.
The DAG of an SCDM can be partitioned into subgraphs, each defining its own SCDM.
Equivalently, this means that two SCDMs can be composed into a new, combined, model.
We use the fact that SCDMs have variables for which no values or probability distributions are assigned -- the roots $\tilde{\mathbf{X}}$.
These variables are the points where the two components can 'attach'.

\begin{definition}[SCDM Composition]
    Consider two SCDMs $$\mathcal{M}_1 = (\mathbf{V}_1, \mathbf{Z}_1, \mathcal{E}_1, (\mathbf{X}_1, \mathbf{D}_1, \mathbf{U}_1), \mathbf{Pr}_1, \mathbf{f}_1, \boldsymbol{\theta}_1)$$
    $$\mathcal{M}_2 = (\mathbf{V}_2, \mathbf{Z}_2, \mathcal{E}_2, (\mathbf{X}_2, \mathbf{D}_2, \mathbf{U}_2), \mathbf{Pr}_2, \mathbf{f}_2, \boldsymbol{\theta}_2)$$ 
    such that $\mathbf{V}_1 \cap \mathbf{V}_2 = \mathbf{Y}_1 \subseteq \tilde{\mathbf{X}}_2$. Then we define the \emph{composition} of the two models $\mathcal{M}_1 \circ \mathcal{M}_2$ as the tuple:

\begin{itemize}
    \item $\mathbf{V}_1 \cup \mathbf{V}_2$
    \item $\mathbf{Z}_1 \cup \mathbf{Z}_2$
    \item $\mathcal{E}_1 \cup \mathcal{E}_2$
    \item $(\mathbf{X}_1 \cup \mathbf{X}_2, \mathbf{D}_1 \cup \mathbf{D}_2, \mathbf{U}_1 \cup \mathbf{U}_2)$
    \item $\mathbf{Pr}_1 \cdot \mathbf{Pr}_2$
    \item $\mathbf{f_1} \cup \mathbf{f_2}$
    \item $\boldsymbol{\theta}_1 \cup \boldsymbol{\theta}_2$
\end{itemize}
    \label{def:composition}
\end{definition}

We will call $\mathbf{Y} = \mathbf{V}_1 \cap \mathbf{V}_2$ the \emph{bridge} of $\mathcal{M}_1 \circ \mathcal{M}_2$. We impose the restriction that for any $Y \in \mathbf{Y}$, then $Y$ is a root variable for either $\mathcal{M}_1$ or $\mathcal{M}_2$.

A composition of two SCDMs is an SCDM.
A \emph{decomposition} of an SCDM is a partitioning of an SCDM into two, such that $\mathcal{M}_0  = \mathcal{M}_1 \circ \mathcal{M}_2 $.
We will refer to this as $\mathbf{Y}$, dropping the subscript, when context makes this unambiguous.
Since bridge variables belong to both components by definition, they can be labeled by either period; in diagrams, we place each variable in the box corresponding to its subscript.

This is a general definition of compositions which does not rely on any assumptions about the model structure. For example, a model can be a composition of two other models that are entirely disjoint.
In the next section, we will focus on a special case of model composition that we will call sequential decomposition.

\subsection{Solving sequentially decomposed SCDMs}

\begin{figure}[h]
\centering
\begin{tikzpicture}[
    node distance=1.4cm and 1.8cm,
    state/.style={circle, draw, minimum size=0.9cm, font=\small},
    decision/.style={rectangle, draw, minimum size=0.9cm, font=\small},
    utility/.style={diamond, draw, minimum size=0.7cm, font=\small, inner sep=1pt},
    bridge/.style={circle, draw, very thick, minimum size=0.9cm, font=\small},
    noise/.style={circle, draw, double, minimum size=0.9cm, font=\small},
    >=stealth
]
\node[decision] (d1) at (0,1) {$\mathbf{D_1}$};

\node[bridge] (y) at (3,0) {$\mathbf{Y}$};

\node[decision] (d2) at (7,1){$\mathbf{D}_2$};

\node[state] (pad2) at (5,1){$\mathbf{Pa_{D_2}}$};

\node[utility] (u2) at (8,-1){$\mathbf{U}_2$};

\draw[->, dashed] (d1) -- (y);


\draw[->, dashed] (y) -- (pad2);

\draw[->] (pad2) -- (d2);

\draw[->, dotted, red] (y) -- (u2);

\draw[->, dashed] (d2) -- (u2);

\begin{scope}[on background layer]
    \node[draw=gray, dashed, rounded corners, fit=(d1)(y), inner sep=0.4cm] (box1) {};
    \node[draw=gray, dashed, rounded corners, fit=(y)(pad2)(d2)(u2), inner sep=0.4cm] (box2) {};
\end{scope}

\node[below=0.3cm of box1.south, font=\footnotesize\itshape] {$\mathcal{M}_1$};
\node[below=0.3cm of box2.south, font=\footnotesize\itshape] {$\mathcal{M}_2$};

\end{tikzpicture}
\caption{A model $\mathcal{M}$ is composed of $\mathcal{M}_1 \circ \mathcal{M}_2$. Indirect paths are represented by dotted edges. If all paths from the bridge nodes $\mathbf{Y}$ to reward nodes in the second component $\mathbf{U_2}$ have a member of $\mathbf{Pa}_{D_2} \cup \mathbf{D_2}$ on it, then $\mathbf{Y}$ is d-separated from $\mathbf{U_2}$ given those nodes. Under that condition the composition is orthomodular or, equivalently, sequential. An indirect path from $Y$ to $U_2$ which is not interrupted by $\mathbf{Pa}_{D_2} \cup \mathbf{D_2}$ (shown in red) breaks this orthomodularity condition.}
\label{fig:sequential-decomposition}
\end{figure}
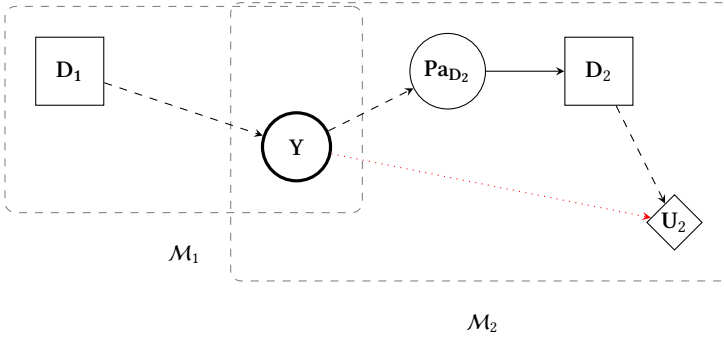

In this section, we discuss solving a decomposed SCIM.
In prior work on influence diagrams, there is a distinction drawn between \emph{perfect recall}, when each decision node has direct information about all past observations and decisions, and \emph{sufficient recall}, when the agent has enough access to prior observations and decisions to make an optimal choice \citep{van2022complete}.
A related notion is that of \emph{strategic reliance} (see Definition \ref{def:strategic-reliance}). Roughly, if $D_2$ does not strategically rely on $D_1$, then the choice of optimal decision rule for $D_2$ does not depend on the choice of optimal decision for $D_1$.
\citet{koller2003multi} show that identifying the network of strategic reliance (see Definition \ref{def:strategic-reliance}) between decision nodes in an influence diagram allows the analyst to divide the problem of finding the optimal policy profile into efficient subgames.
In this section, we will consider a special case of sequential decomposition and how it affords the breaking down of the model's solution problem into smaller subproblems.

Consider  $\mathcal{M}_0 = \mathcal{M}_1 \circ \mathcal{M}_2$
In order to get computational power out of decomposing an SCIM, the decomposition has to separate the decision variables such that $\mathbf{D_2}$ does not strategically rely on $\mathbf{D_1}$.
Under these conditions, the optimal decision rules for $\mathbf{D_2}$ will not depend on the decision rules for $\mathbf{D_1}$, and as we will see this enables the problem to be solved through a simple backwards induction procedure.

We can guarantee this by imposing a graphical criterion on the decomposition. It relies on the definition of d-separation (see Definition \ref{def:d-separation}).

\begin{definition}[Orthomodularity]
Given $\mathcal{M}_0 = \mathcal{M}_1 \circ \mathcal{M}_2$  with bridge $\mathbf{Y}$, the decomposition is \emph{orthomodular} if and only if $\mathbf{Y}$ and $\mathbf{U_2}$ are d-separated given $\mathbf{D_2} \cup  \mathbf{Pa_{D_2}}$.
\label{def:orthomodular}
\end{definition}

Note that we are presenting orthomodularity as a sufficient, but not necessary, condition for the separability of component subgames of an SCDM.

\begin{theorem}
Given an orthomodular decomposition $\mathcal{M}_0 = \mathcal{M}_1 \circ \mathcal{M}_2$, 
then there is no $D_2 \in \mathbf{D_2}$ that strategically relies on $D_1 \in \mathbf{D_1}$.
\label{thm:non-reliant-decomposition}
\end{theorem}
\begin{proof}
By Definition \ref{def:composition}, all paths between $\mathbf{D_1}$ and $\mathbf{U_{D_2}}$ must include at least one bridge node $Y$, and $\mathbf{D_1} \cap \mathbf{V_2}= \emptyset$.

All paths from $\hat{\mathbf{D_1}}$, added parents of $\mathbf{D_1}$, to $\mathbf{U_{D_2}}$ will include a bridge node $Y \in \mathbf{Y}$.

By Definition \ref{def:orthomodular}, $\mathbf{Y}$ and $U_2$ are d-separated given $\mathbf{D_2} \cup  \mathbf{Pa_{D_2}}$.

Therefore there are no active paths from  $\mathbf{D_1}$ to $U_2$ given $\mathbf{D_2} \cup  \mathbf{Pa_{D_2}}$.

By Definition \ref{def:s-reachability}, $\mathbf{D_1}$ is not s-reachable from $\mathbf{D_2}$.

By Theorem \ref{thm:km-soundness}, $\mathbf{D_2}$ does not strategically rely on $\mathbf{D_1}$.
\end{proof}

We now revisit the problem of solving for the optimal policy profile of an SCIM given the tool of sequential decomposition.

\begin{definition}[Bridge value function]
Given an orthomodular decomposition $\mathcal{M}_0 = \mathcal{M}_1 \circ \mathcal{M}_2$, the \emph{bridge value function} of $\mathcal{M}_2$  is:

\begin{equation}
    v_{\mathcal{M}_2}(\mathbf{y}) = \max_{\boldsymbol{\pi}_2} E_\mathbf{Z} \left[ \sum_{U \in \mathbf{U}_2} f^*_{U}(\mathbf{y}, \mathbf{z}, \boldsymbol{\pi_2}) \right]
    \label{eq:decomposed-m2-root-value}
\end{equation}
\end{definition}

By Theorem \ref{thm:non-reliant-decomposition}, we know that in the optimal profile policy for the composed model $\mathcal{M}_0$, the optimal policy profile for $\pi_2$ does not depend on the choice of $\pi_1$. This means that we can use the bridge value function as a \emph{continuation value function} in the solution for the total model.

\begin{equation}
\begin{split}
 v_{\mathcal{M}_0}(\tilde{\mathbf{x}}) & = \max_{\boldsymbol{\pi}_0} E_{\mathbf{Z}_0} \left[ \sum_{U \in \mathbf{U_0}} f^*_{U}(\tilde{\mathbf{x}}, \mathbf{z_0}, \boldsymbol{\pi}) \right] \\
 &=  \max_{\boldsymbol{\pi}_1} E_{\mathbf{Z}_1} \left[ \sum_{U_1 \in \mathbf{U_1}} f^*_{U_1} (\tilde{\mathbf{x}}, \mathbf{z_1}, \boldsymbol{\pi_1}) + \max_{\boldsymbol{\pi}_2} E_{\mathbf{Z}_2} \left( \sum_{U_2 \in \mathbf{U}_2} f^*_{U_2}(f_\mathbf{y}(\tilde{\mathbf{x}}, \mathbf{z_1}, \boldsymbol{\pi_1}), \mathbf{z_2}, \boldsymbol{\pi_2} \right) \right]\\
 & = \max_{\boldsymbol{\pi}_1} E_{\mathbf{Z}_1} \left[ \sum_{U_1 \in \mathbf{U_1}} f^*_{U_1} (\tilde{\mathbf{x}}, \mathbf{z_1}, \boldsymbol{\pi_1}) + v_{\mathcal{M}_2}(f_\mathbf{y}(\tilde{\mathbf{x}}, \mathbf{z_1}, \boldsymbol{\pi_1}), \boldsymbol{\pi_1}) \right]
\end{split}
\label{eq:decomposed-m0-value}
\end{equation}

\subsection{Computational benefits of decomposition}
\label{sec:computational-benefit-static}

Consider the problem of finding the optimal policy profile $\mathbf{\pi}_0$ for model $\mathcal{M}_0 = \mathcal{M}_1 \circ \mathcal{M}_2$, with orthomodular decomposition.
Let $|\mathbf{\pi}_i|$ be the number of possible total policy profiles in model $\mathcal{M}_i$. $|\mathbf{\pi}_0|= |\mathbf{\pi}_1||\mathbf{\pi}_2|$.
Searching for the optimal $\mathbf{\pi}_0$ in the most naive way involves a search over $|\mathbf{\pi}_0|$ possibilities that is $O(|\mathbf{\pi}_0|)$.

Because the model is sequential, the subgame $\mathcal{M}_2$ can be solved for every case of the bridge $\mathbf{Y}$ with a worst-case search.
Once solved, the bridge value function can be solved in constant time $O(1)$ (assuming the structural equations $\mathbf{f}$ are all $O(1)$).
Thus, solving the decomposed model in sequence can be done in $O(|\pi_2||Y| + |\pi_1|)$. This is a significant improvement.

\subsection{Example: Two-period consumption with habit formation}
\label{sec:habit-formation}

Having seen how the standard two-period consumption problem decomposes sequentially, we now consider a variation that \emph{requires an expanded bridge}: habit formation. This example demonstrates how certain structural features necessitate richer sufficient statistics for decomposition.

Consider a consumer who lives for two periods and seeks to maximize lifetime utility:
$$\max_{c_1, c_2} u(c_1, h_1) + \beta u(c_2, h_2),$$
where $c_t$ denotes consumption in period $t$, $h_t$ denotes the habit stock in period $t$, $u(\cdot, \cdot)$ is a utility function that depends on both consumption and habits, and $\beta \in (0,1)$ is a time preference factor.

The term ``habit formation'' here refers not to behavioral routines but to preferences in which past consumption affects current utility. Utility is decreasing in the habit stock (that is, $\partial u / \partial h < 0$): higher past consumption raises the reference point against which current consumption is compared, making any given level of current consumption less satisfying.

The budget dynamics are identical to the standard problem:
$$a_1 = b_1 + y_1 - c_1,$$
$$b_2 = R \cdot a_1,$$
$$c_2 = b_2 + y_2.$$

However, we now add an equation describing how habits evolve. In the simplest formulation:
$$h_2 = c_1.$$
That is, the habit stock in period 2 equals consumption in period 1. We take the initial habit $h_1$ as an exogenous parameter reflecting the consumer's consumption history prior to period 1; it enters the period 1 utility function and, as we will see, becomes part of the state space for the optimization problem.

\paragraph{Why the bridge must expand} The critical difference from the standard model lies in how information flows between periods. In the standard problem, period 2 utility $u(c_2)$ depends only on $c_2$, which is determined by the bridge variable $a_1$ (equivalently $b_2$) and period 2 income $y_2$. The single bridge variable $a_1$ serves as a sufficient statistic: it captures all the information from period 1 that is relevant for period 2.

With habit formation, period 2 utility $u(c_2, h_2)$ depends on both $c_2$ and $h_2 = c_1$. The consumption choice $c_1$ now affects period 2 utility through two distinct channels: (i) indirectly through $a_1 \to b_2 \to c_2$, and (ii) directly through $h_2 = c_1$. This means that $a_1$ alone is no longer a sufficient statistic; period 2 outcomes depend on information from period 1 that is not captured by $a_1$.

The solution is to \emph{expand the bridge} to include both channels. If we take $\mathbf{Y} = \{a_1, c_1\}$ as the bridge (or equivalently $\{b_2, h_2\}$ on the period 2 side), then all connections between periods pass through the bridge. With this expanded bridge, $c_2$ becomes a genuine decision variable: the consumer in period 2 observes both their wealth $b_2$ and their habit stock $h_2$, and chooses consumption accordingly.

Under this formulation, the orthomodularity condition of Definition~\ref{def:orthomodular} is satisfied. With $c_2$ as a decision variable depending on $h_2$, the path $c_1 \to h_2 \to u_2$ is blocked: $h_2$ is now a parent of the decision $c_2$, so conditioning on $\mathbf{D}_2 \cup \mathbf{Pa}_{\mathbf{D}_2} = \{c_2, b_2, y_2, h_2\}$ includes $h_2$, which blocks the path. The expanded bridge $\{b_2, h_2\}$ and $\mathbf{U}_2$ are d-separated given the period 2 decision and its parents, restoring the backwards induction strategy of Theorem~\ref{thm:non-reliant-decomposition}.

\begin{figure}[h]
\centering
\begin{tikzpicture}[
    state/.style={circle, draw, minimum size=0.9cm, font=\small},
    decision/.style={rectangle, draw, minimum size=0.9cm, font=\small},
    utility/.style={diamond, draw, minimum size=0.7cm, font=\small, inner sep=1pt},
    noise/.style={circle, draw, double, minimum size=0.9cm, font=\small},
    discount/.style={regular polygon, regular polygon sides=5, draw, minimum size=1.1cm, font=\small, inner sep=1pt},
    >=stealth
]

\node[state] (b1) at (0, 0) {$b_1$};
\node[state] (a1) at (3.5, 0) {$a_1$};

\node[state] (y1) at (0, -2) {$y_1$};

\node[decision] (c1) at (1.5, -4) {$c_1$};
\node[utility] (u1) at (3, -5) {$u_1$};

\node[state] (h1) at (0, -5) {$h_1$};

\node[state] (b2) at (6, 0) {$b_2$};
\node[state] (h2) at (6, -5) {$h_2$};

\node[state] (R) at (4, -1) {$R$};

\node[state] (y2) at (8, -2) {$y_2$};

\node[decision] (c2) at (8.5, -4) {$c_2$};
\node[utility] (u2) at (10.5, -5) {$u_2$};

\node[discount] (beta) at (10, -3) {$\beta$};

\draw[->] (b1) -- (c1);
\draw[->] (y1) -- (c1);

\draw[->] (b1) -- (a1);
\draw[->] (y1) -- (a1);
\draw[->] (c1) -- (a1);

\draw[->] (c1) -- (u1);
\draw[->] (h1) -- (u1);

\draw[->] (a1) -- (b2);
\draw[->] (R) -- (b2);

\draw[->, dashed] (c1) -- (h2);

\draw[->] (b2) -- (c2);
\draw[->] (y2) -- (c2);
\draw[->] (h2) -- (c2);

\draw[->] (c2) -- (u2);
\draw[->] (h2) -- (u2);

\draw[->, dashed] (beta) -- (u2);

\begin{scope}[on background layer]
    \node[draw=gray, dashed, rounded corners, fit=(b1)(y1)(c1)(u1)(a1)(h1)(R)(b2)(h2), inner sep=0.4cm] (box1) {};
    \node[draw=gray, dashed, rounded corners, fit=(b2)(h2)(y2)(c2)(u2)(beta), inner sep=0.4cm] (box2) {};
\end{scope}

\node[below=0.3cm of box1.south, font=\footnotesize\itshape] {Period 1};
\node[below=0.3cm of box2.south, font=\footnotesize\itshape] {Period 2};

\end{tikzpicture}
\caption{Influence diagram for the two-period consumption problem with habit formation. The dashed arrow from $c_1$ to $h_2$ represents the habit formation mechanism: period 1 consumption directly affects period 2 habits. The expanded bridge $\{b_2, h_2\}$ connects the two periods; the boxes overlap on these nodes to indicate they belong to both components. With this expanded bridge, $c_2$ becomes a genuine decision variable (rectangle) that depends on both wealth $b_2$ and habits $h_2$. Pentagons denote discount factors.}
\label{fig:two-period-habit}
\end{figure}
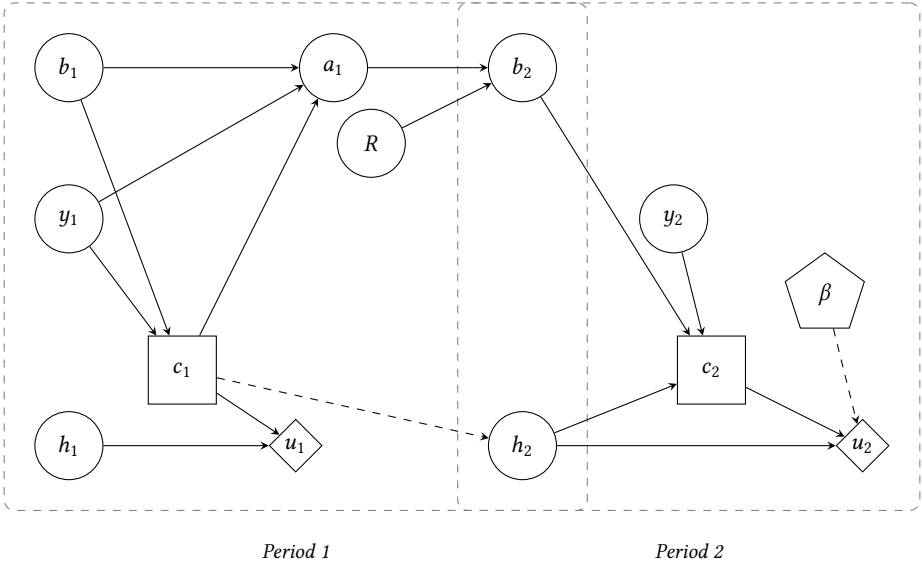

\paragraph{Implications for solution methods} The expanded bridge has practical consequences for solution methods. In the standard problem, we evaluate period 2 for each possible value of the single bridge variable $a_1$, obtaining a one-dimensional continuation value function $v_2(b_2)$. With habits, the continuation value function is two-dimensional, defined over the expanded bridge:
$$v_2(b_2, h_2) = \max_{c_2} u(c_2, h_2) \quad \text{subject to } c_2 \leq b_2 + y_2.$$
The consumer in period 2 now makes a genuine choice: given wealth $b_2$ and habit stock $h_2$, they choose $c_2$ to maximize utility. This continuation value function can be computed independently of period 1, and then used in the period 1 problem:
\begin{equation}
v_1(b_1, h_1) = \max_{c_1} \left\{ u(c_1, h_1) + \beta \, v_2\bigl(R(b_1 + y_1 - c_1), c_1\bigr) \right\}.
\label{eq:habit-bellman}
\end{equation}
The first argument of $v_2$ is the wealth channel (via $a_1 \to b_2$); the second is the habit channel (via $c_1 \to h_2$). The problem decomposes sequentially, but over a higher-dimensional bridge than the standard model.

\paragraph{General lessons} This example illustrates a broader point about model structure and decomposability. The standard consumption problem decomposes with a single bridge variable because asset position is a sufficient statistic for all period 1 information relevant to period 2. Habit formation requires expanding the bridge: past consumption matters for future utility beyond its effect on wealth, so both $a_1$ and $c_1$ (or equivalently $b_2$ and $h_2$) must be included. More generally, when past decisions affect future outcomes through multiple channels, the bridge must expand to capture all relevant sufficient statistics. The framework accommodates such cases: decomposition does not fail, but the dimensionality of the bridge increases.

\section{Structural Causal Decision Processes (SCDP)}
\label{sec:scdp}

We now build on the prior definitions of SCDMs and composition to develop a new form of composable model for a dynamic decision process, an SCDP.
Like the SCDMs, a decomposable SCDP can be solved more efficiently through separation into subgames.

\subsection{Example: Consumer choice with portfolio allocation}
\label{sec:consumer-with-portfolio}

We consider two components of a dynamic stochastic optimization problem.
A consumer earns and consumes, anticipating the availability of resources for future consumption.
They also make a portfolio allocation decision which exposes their savings to a more or less risky rate of return.
This model is diagramed in Figure \ref{fig:consumer-portfolio}.

\paragraph{Consumption-saving problem}  Consider a simple consumption-saving problem with borrowing constraint \citep{maliar2021deep}. The consumer aims to maximize, through choice of their level of consumption $c_t$:
$$\max_{\{c_t\}} E \left[ \sum_{t=0}^{\infty} \beta^t u (c_t) \right]$$

subject to the following transition equations and constraints:

\begin{align}
    y &= \mu_y + \sigma_y \epsilon_y ;\ \epsilon_y \sim \mathcal{N}(0,1)\\
    m &= w + e^{y}\\
    0 \leq c(m) &\leq m\\
    u &= \text{ln } c\\
    a &= m - c\\
    w' &= ra
\end{align}

where $w$ is the level of wealth, $r$ is a rate of return on savings, and $y$ is an income shock, and $u$ is utility from consumption.

\paragraph{Portfolio allocation} We can make the consumption-saving problem more complex by introducing a risky return and portfolio allocation choice $\alpha$.
Whereas in the earlier problem the rate of return $r$ was given, now it is an endogenously scaled mixture of a constant $r_{f}$ and a stochastic 
shock $r_{r}$. The consumer chooses $\alpha$, the share of their wealth to invest in the risky asset.

\begin{align}
0 &\leq \alpha(a) \leq 1\\
r_{r} &= e^{\mu_r + \sigma_r \epsilon_r} ;\ \epsilon_r \sim \mathcal{N}(0,1)\\
r &= (1 - \alpha) r_{f} + \alpha r_{r}
\end{align}

A critical piece of the interpretation of these equations is that the portfolio allocation choice $\alpha$ is made with knowledge of the post-consumption wealth $a$ but without the information of the realization of $r_{r}$.

\paragraph{Composed model}. The composed model is shown in Figure \ref{fig:consumer-portfolio}.
The variables $a$ and $r$ are shared by both models and thus are the \emph{bridge}.
Because the allocation component has no utility nodes, the composition is trivially orthomodular.
The single period problem can be solved first by solving for $\alpha$ given the continuation value of $w'$, and then by solving for optimal $c$.

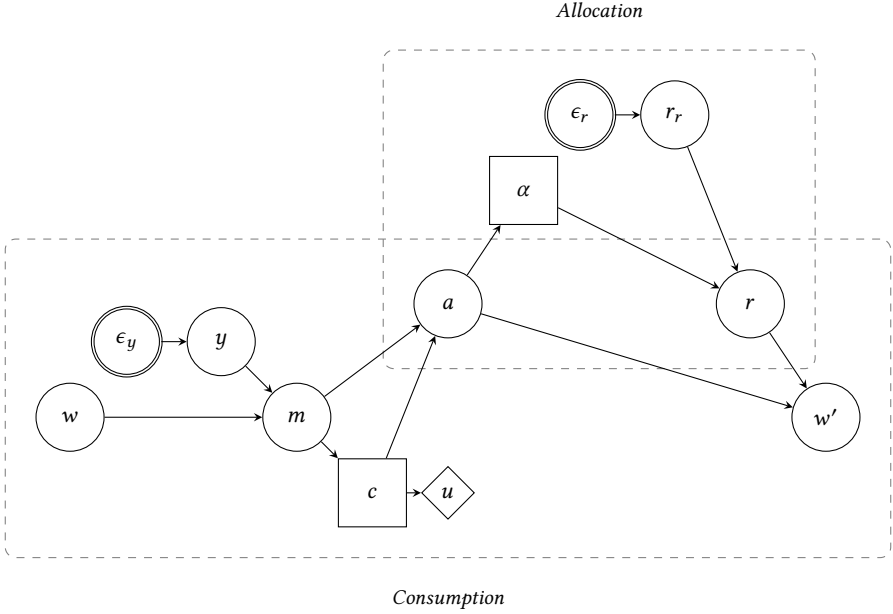
\begin{figure}[h]
\centering
\begin{tikzpicture}[
    node distance=1.4cm and 1.8cm,
    state/.style={circle, draw, minimum size=0.9cm, font=\small},
    decision/.style={rectangle, draw, minimum size=0.9cm, font=\small},
    utility/.style={diamond, draw, minimum size=0.7cm, font=\small, inner sep=1pt},
    bridge/.style={circle, draw, very thick, minimum size=0.9cm, font=\small},
    noise/.style={circle, draw, double, minimum size=0.9cm, font=\small},
    parameter/.style={draw=none, font=\small},
    >=stealth
]

\node[state] (w) at (-2,0) {$w$};
\node[noise] (eps) at (-1.25,1) {$\epsilon_y$};
\node[state] (y) at (0,1) {$y$};
\node[state] (m) at (1,0) {$m$};

\node[decision] (c) at (2,-1) {$c$};
\node[utility] (u) at (3,-1) {$u$};
\node[state] (a) at (3,1.5) {$a$};

\node[decision] (alpha) at (4,3) {$\alpha$};
\node[noise] (eps_2) at (4.75, 4) {$\epsilon_r$};
\node[state] (r_r) at (6,4) {$r_r$};
\node[state] (r) at (7,1.5) {$r$};

\node[state] (w') at (8,0) {$w'$};

\draw[->] (eps) -- (y);
\draw[->] (y) -- (m);
\draw[->] (w) -- (m);
\draw[->] (m) -- (c);
\draw[->] (c) -- (u);
\draw[->] (m) -- (a);
\draw[->] (c) -- (a);
\draw[->] (a) -- (w');
\draw[->] (r) -- (w');

\draw[->] (a) -- (alpha);
\draw[->] (alpha) -- (r);
\draw[->] (r_r) -- (r);
\draw[->] (eps_2) -- (r_r);

\begin{scope}[on background layer]
    \node[draw=gray, dashed, rounded corners, fit=(w)(eps)(u)(m)(a)(y)(c)(w'), inner sep=0.4cm] (box1) {};
    \node[draw=gray, dashed, rounded corners, fit=(a)(alpha)(r_r)(eps_2)(r), inner sep=0.4cm] (box2) {};
\end{scope}

\node[below=0.3cm of box1.south, font=\footnotesize\itshape] {Consumption};
\node[above=0.3cm of box2.north, font=\footnotesize\itshape] {Allocation};

\end{tikzpicture}
\caption{The composed consumption and portfolio allocation dynamic problem.}
\label{fig:consumer-portfolio}
\end{figure}

\subsection{Dynamic Optimization}

We begin with a simple definition of a dynamic stochastic optimization problem, drawn from \cite{maliar2021deep}.

\begin{definition}[Dynamic stochastic optimization problem (DSOP)] Given discrete time steps $t \in {0, 1, ..., T}$ with $T$ potentially infinite, and given exogenous shocks $z = \{z_t \sim P_{Z_t}\}^T_{t=0}$ and endogenous state $x = \{x_t\}^T_{t=0}$ and actions $a = \{a_t\}^T_{t=0}$ such that:
\begin{itemize}
    \item $a_t \in A_t(x_t, z_t)$
    \item $x_{t+1} = g_t(x_t, z_t, a_t)$
\end{itemize}
    the agent maximizes expected lifetime value
    \begin{equation}
v_0 = \max_{a,x} E_{z} \left[\sum_{t=0}^{T-1}\beta^t r(x_t, z_t, a_t) \right]
\label{eq:lifetime-value}
\end{equation}
where $\beta \in [0,1)$ is the discount factor.
\label{def:dsop}
\end{definition}

As is well-known, this sort of optimization problem can be represented by a recursive Bellman equation.
We will use a somewhat adjusted form.

\begin{definition}[Shifted Bellman Equations]
The Shifted Bellman Equation is for $v_t(x_t)$. We let $v_{T}(x_t) = 0$. For $t < T$,
\begin{equation}
    v_t(x_t) = E_{z_t} \left[ \max_{a_t \in A_t(x_t,z_t)} \{ r_t(x_t,z_t,a_t) + \beta_t v_{t+1}(x_{t+1}) \} \right]
    \label{eq:shifted-bellman-equation}
\end{equation}
where $x_{t+1} = g_t(x_t,z_t,a_t)$.
\end{definition}

The benefit of Equation \ref{eq:shifted-bellman-equation} is that it formulates the ``solution'' of a single time period $t$ in terms of a small number of self-contained elements.
We will give this bundle of elements a name -- a period block, or P-block.

\begin{definition}[P-block]
A P-block is a tuple $\mathcal{B}^P_t = (\mathcal{X}, \mathcal{Z}, P_\mathcal{Z}, \mathcal{A}, A, g, r, \beta_t)$:
\begin{itemize}
    \item A state space $\mathcal{X}$.
    \item A space of exogenous shocks $\mathcal{Z}$.
    \item A probability distribution over shocks $P_\mathcal{Z}$.
    \item A space of actions $\mathcal{A}$
    \item A set of action constraints $A: \mathcal{X} \times \mathcal{Z} \rightarrow \mathbb{P}(\mathcal{A})$
    \item A deterministic transition function $g : \mathcal{X} \times \mathcal{Z} \times \mathcal{A} \rightarrow \mathcal{X}$
    \item A reward function $r:  \mathcal{X} \times \mathcal{Z} \times \mathcal{A} \rightarrow \mathbb{R}$
    \item A discount function $\beta: \mathcal{X} \times \mathcal{Z} \times \mathcal{A} \rightarrow \mathbb{R}$, which is most often valued as a constant.
\end{itemize}
\label{def:p-block}
\end{definition}

We have allowed a functional or variable discount factor here.
We will explore what the implications of that are in Section \ref{sec:scdp-variable-discount}.
For purposes of presentation here, we will continue to show $\beta$ as a constant.


\begin{lemma} Given, for $t = \{0, 1, ..., T - 1\}$, P-Block $\mathcal{B}^P_t$ 
then the optimization problem in shifted Bellman* form
\begin{equation}
    v_t(x_t) = E_{z_t} \left[ \max_{a_t \in A_t(x_t,z_t)} \{ r_t(x_t,z_t,a_t) + \beta_t v_t(g(x_t, z_t, a_t)) \} \right]
    \label{eq:t_block-bellman-equation}
\end{equation} is a DSOP.
\end{lemma}
\begin{proof}
Equation \ref{eq:t_block-bellman-equation} implies Equation \ref{eq:shifted-bellman-equation}, which is an alternative form for a DSOP.
\end{proof}

\subsubsection{Equivalence of SCIMs and P-blocks}

We have defined two mathematical objects, a P-block (\ref{def:p-block}) which represents a single 'period' of a DSOP (\ref{def:dsop}), and an SCDM (\ref{def:scdm}). While these two constructs come from different intellectual fields, it is simple to produce a P-block from an SCDM.

\begin{theorem} Given
\begin{itemize}
    \item an SCDM $\mathcal{M} = (\mathbf{V}, \mathbf{Z}, \mathcal{E}, (\mathbf{X}, \mathbf{D}, \mathbf{U}), \mathbf{Pr}, \mathbf{f}, \boldsymbol{\theta})$
    \item a discount variable $\beta \in \mathbf{V}$
    \item a mapping from root variables to end-of-period-variables
    $w: \mathbf{\tilde{X}} \rightarrow \mathbf{V}$
    such that
    $\forall \tilde{X} \in \mathbf{\tilde{X}}, \text{dom}(w(\tilde{X})) = \text{dom}(\tilde{X}))$ and 
    $\mathbf{W} = \{w(\tilde{X}) | \tilde{X} \in \mathbf{\tilde{X}} \}$
\end{itemize} there is a corresponding P-block $\mathcal{B}^P_t = (\mathcal{X}, \mathcal{Z}, P_\mathcal{Z}, \mathcal{A}, A, g, r, \beta_t)$. 
\end{theorem}

\begin{proof}
    By construction.
    
    Let $\mathcal{X} \subseteq \tilde{\mathbf{X}}$, the roots.
    
    Let $\mathcal{Z} = \mathbf{Z}$.

    Let $P_\mathcal{Z} = Pr(\mathbf{Z}, \mathbf{\theta})$

    Let $\mathcal{A} = \mathbf{D}$.


    Let $A = \Gamma$.

    Let $g(x, z, a) = \mathbf{f}_{\mathbf{W}}(x, z, a)$, the structural functions evaluated at the end-of-state variables $\mathbf{W}$.

    Let $r(x, z, a) = \sum_{U \in \mathbf{U}} f^*_U(x, z, a)$, the sum of utility functions.

    Let $\beta = f^*_\beta$.
\end{proof}

Thus, an SCIM can be made 'dynamic' by adding a discount factor and by making an explicit mapping from the available model variables to the state variables for the next period.
We now have a method of creating dynamic stochastic optimization problems from an SCDM. We will call these \emph{Structural Causal Decision Processes} (SCDP).

\begin{definition}[Structural Causal Decision Process] A \emph{structural causal decision process} $(\mathcal{M}, \beta, w)$ is a SCIM $\mathcal{M} = (\mathbf{V}, \mathbf{Z}, \mathcal{E}, (\mathbf{X}, \mathbf{D}, \mathbf{U}), \mathbf{Pr}, \mathbf{f}, \boldsymbol{\theta})$, a discount variable $\beta \in \mathbf{V}$, and an end-state mapping $w: \mathbf{\tilde{X}} \rightarrow \mathbf{V}$, composed, via P-block, into a DSOP.
\label{def:scdp}
\end{definition}


\subsubsection{Dynamic Models and Sequential Decomposition}
\label{sec:composition}

We have shown how sequential decomposition reduces the complexity of solving static models in Section \ref{sec:computational-benefit-static}.
Here, we will show that these same computational advantages carry through to the dynamic case.

Consider a P-block $\mathcal{B}^P$ constructed from an SCDM which sequentially decomposes $\mathcal{M}_0 = \mathcal{M}_1 \circ \mathcal{M}_2$, and an algorithm that attempts to derive the value function by iteratively recomputing the value function over all states $\mathcal{X}$, taking expectation over all shocks $\mathcal{Z}$, and optimizing over all actions $\mathcal{A}$. Naively, this results in an $O(|\mathcal{X}||\mathcal{Z}||\mathcal{A}|)$ step for each iteration.

However, because the SCDM is sequentially decomposable, it is possible to solve the optimization problem embedded in the update more efficiently. Whereas $|\mathcal{A}| = |\boldsymbol{\pi}_0|$, we have seen that the second subgame can be solved separately from the first.
The complexity of the optimization is $O(|\boldsymbol{\pi}_2||\mathbf{Y}| + |\boldsymbol{\pi}_1|)$, resulting in a total complexity for an iteration of $O(|\mathbf{\tilde{X}}||\mathcal{Z}|(|\boldsymbol{\pi}_2||\mathbf{Y}| + |\boldsymbol{\pi}_1|))$, a significantly tighter bound.

\section{SCDPs beyond POMDPs toward resource rationality}
\label{sec:scdp-pomdp-rr}

We will discuss the relationship between an SCDP and more commonly known formalisms, MDP and POMDP.
While an SCDP has many elements in common with MDPs and POMDPs, it is not exactly the same thing as these.
The class of SCDPs is not the same as the class of MDPs because in an SCDP, it is possible that the agent will have limited information when making their decisions.
The class of SCDPs is not the same as the class of POMDPs because it is not guaranteed that the agent in the SCDP has memory of sufficient statistics with which to solve for optimal behavior.
Rather, the SCDP framework is flexible enough to model the memory capacity of agents directly.

\subsection{Example: latent income process}
\label{sec:scdp-example}

Figure \ref{fig:example-scdp} depicts a valid SCDP, governed by the following equations:
\begin{align*}
    \epsilon_a &\sim Normal(0,1)\\
    a' &= a + e^{\epsilon_a}\\
    c &= b + a\\
    0 &< d(c) < c\\
    \epsilon_b &\sim Normal(0,1)\\
    b' &= c - d + {e^\epsilon_b}\\
    u &= \text{log}\ d
\end{align*}

This is similar to the consumer problems we have discussed before, except now the agent is subject to an unobserved exogenous income process $(a, \epsilon, a')$.
Each period, they experience both a transitory and a permanent (evolving over time) income shock, but they do not know whether the income is from the transitory or permanent source.
If the agent were able to observe $b$ at $d$, their decision rule would in effect have full information of $a$, but this is not the case.

Given these equations, the agent has no past memory with which to formulate beliefs about their unobserved environment. Rather, they can infer the latent value of $a$ only through their observation $c$ directly. Any inferences about the environment will be implicit in the decision rule $\pi_d$.

However, if we add variables and connections, we can construct the capacity for the agent to have a Bayesian filter with which to retain past information and make better decisions.

\begin{align*}
    m'(m, c) &\in \{0,1\}\\
    0 < d(c, m') &< c
\end{align*}

In this extended model, the agent chooses two decision rules. $\pi_d$ governs the consumption decision. $\pi_{m'}$ governs the retention of information from observations $c$. We would expect that the agent equipped with this additional memory unit would be able to capture greater lifetime reward. But for the sake of exposition, we have restricted the 'belief states' of this agent to only two states, $0$ or $1$.

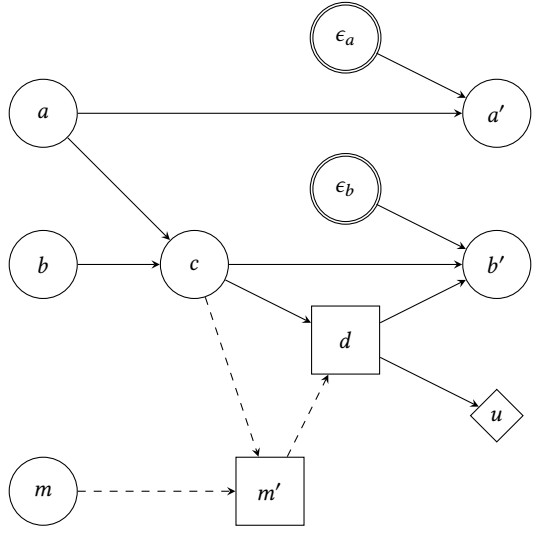
\begin{figure}
\centering
\begin{tikzpicture}[
    state/.style={circle, draw, minimum size=0.9cm, font=\small},
    decision/.style={rectangle, draw, minimum size=0.9cm, font=\small},
    utility/.style={diamond, draw, minimum size=0.7cm, font=\small, inner sep=1pt},
    bridge/.style={circle, draw, very thick, minimum size=0.9cm, font=\small},
    noise/.style={circle, draw, double, minimum size=0.9cm, font=\small},
    >=stealth
]

\node[state] (a) at (0,0) {$a$};
\node[state] (a') at (6,0) {$a'$};
\node[noise] (eps_a) at (4,1) {$\epsilon_a$};
\node[state] (b) at (0, -2) {$b$};
\node[state] (c) at (2, -2) {$c$};
\node[decision] (d) at (4, -3) {$d$};
\node[state] (b') at (6, -2) {$b'$};
\node[noise] (eps_b) at (4,-1) {$\epsilon_b$};
\node[utility] (u) at (6, -4) {$u$};
\node[state] (m) at (0, -5) {$m$};
\node[decision] (m') at (3, -5) {$m'$};

\draw[->] (a) -> (a');
\draw[->] (eps_a) -> (a');
\draw[->] (b) -> (c);
\draw[->] (a) -> (c);
\draw[->] (c) -> (d);
\draw[->] (d) -> (b');
\draw[->] (c) -> (b');
\draw[->] (eps_b) -> (b');
\draw[->] (d) -> (u);

\draw[->, dashed] (m) -> (m');
\draw[->, dashed] (c) -> (m');
\draw[->, dashed] (m') -> (d);

\end{tikzpicture}
\caption{An SCDP with latent state and shocks. The agent observes $c$ at decision $d$, but not the exogenous income process $(a, \epsilon_a, a', \epsilon_b)$. Double circles denote exogenous noise variables. An SCDP need not have a Bayesian filter, which is typically considered part of a POMDP, but it is possible to build one in. The variables $m, m'$ and connecting dashed edges represent that (optional) filter.}
\label{fig:example-scdp}
\end{figure}

\subsection{SCDP and MDP}

An MDP is a tuple $(S, A, P_a, R_a)$, where $S$ are the states, $A_s$ are the actions available at each state $s \in S$, $P_a(s, s')$ are the transition probabilities, and $R(s,a)$ is the immediate reward after having taken action $a$ in state $s$. There are many variations on this formalism. This is very similar to the definition of the DSOP in Definition \ref{def:dsop}, in that it sets up all the elements needed for the dynamic optimization problem, but for the discount factor $\beta$.

However, canonically, the policy function in an MDP is $\pi(s)$, a function of state values.
This implies that the agent has full information about the state variables when they make their decision -- it is a 'fully observed' model.

In contrast, in an SCDP, decision rules $\pi_D$ are functions of the $\mathbf{Pa}_D$, which may not include all the variables in $\mathbf{\tilde{X}}$.
SCDPs are expressive enough to allow for partial-observability or latent variables in a model.
That includes the case in which decisions are made sequentially but with exogenous shocks interleaved in between them.

\subsection{SCDP and POMDP}

POMDPs are an extension to MDPs that separate states from observations \cite{krishnamurthy2016partially}.
It is a tuple $(S, A, \Omega, T, O, R, \beta)$, where $S$ are states, $A$ are actions, $\Omega$ are a set of possible observations, $T(s' | s, a)$ is the transition probability distribution, $O(o|s', a)$ is the observation probability distribution, $R(s,a)$ is the reward function, and $\beta$ is the discount factor.

POMDPs are typically analyzed as implying an internal belief state MDP $(B, A, \tau, r, \beta)$, which is a kind of subjective mirror to the real environment described by the POMDP. A belief transition function $\tau(b'|a,b)$ operates as a Bayesian filter, aggregating information from prior observations and storing what is ideally a sufficient statistic in the belief state $b$, and $r(a,b)$ computes the expectations of reward $R$ given belief in $b$.
The agent 'solves' the POMDP by constructing and solving the belief state MDP.

This formation of the belief state MDP implies that there is an information flow from past beliefs to the information set of current decisions which is \emph{not required} for an SCDP. The SCDP makes explicit which variables are in the information sets for its decision variable(s).
For example, in the first version of the model we introduce in Section \ref{sec:scdp-example}, the agent does not remember prior information.
They are limited to their immediate observation $c$ but have no way of tracking the latent state beyond that signal.
In the second version, the filter $m$ serves as the memory for the agent, subject to the agents' decision rule $\pi_m$.

So, SCDPs are a more general form of model than POMDPs.
There are other ways in which SCDPs are more expressive than POMDPs.
For example, a single SCDP can have multiple decision variables, each with a different information set, with different linking 'memory' states.

\subsection{Expressing resource rational and other realistic agents}

Our contention is that this additional expressiveness that SCDPs have, beyond the POMDP, is useful when considering how to model realistic settings with agents.
This is because the rationality assumptions implied by the POMDP model are not realistic for many social scientifically valid agents that act with limited information, irrational decision rules, and operate in complex environments that are composed out of many contexts.
(See the discussion of resource rationality in Section \ref{sec:resource-rationality}.)

Minor variations of the example model can expand the memory available to the agent ($m \in \mathbb{R})$, change the way that memory works (by making $m'(m,x)$ a fixed function rather than a flexible decision rule), or by composing it with additional states (as in the portfolio allocation module in Section \ref{sec:consumer-with-portfolio}).
We do not guarantee that all environments developed in this way are efficiently soluble. Rather, we are concerned with how to flexibly describe agent environments based on their realistic informational and computational constraints.

\subsection{Example: costly memory}

We present one more example that illustrates how SCDPs can represent the design of artificial agents in a way that exceeds the representational power of POMDPs.
It varies the model in Section \ref{sec:scdp-example} only slightly by allowing the agent to make a decision about \emph{how well they remember} past information.
This is a model of an agent which can retain data about past observations, but at a cost. This reflects the contemporary situation of firms which may have to pay for cloud storage for their data.

We retain the equations of the consumer with a latent income process.
\begin{align*}
    \epsilon_a &\sim Normal(0,1)\\
    a' &= a + e^{\epsilon_a}\\
    c &= b + a\\
    0 < d(c, q) &< c\\
    \epsilon_b &\sim Normal(0,1)\\
    b' &= c - d + {e^\epsilon_b}\\
    u &= \text{log}\ d
\end{align*}

However, we alter the memory process $(p,q)$ so that it is noisy $\epsilon_p$, with the noise modulated by a tunable rate of remembering $r$ which introduces a cost $v$.
\begin{align*}
    q(p, c) &\in \mathbb{R}\\
    r &> 0\\
    v &= - 1/r\\
    \epsilon_p &\sim Normal(0,1)\\
    p' &= q + \epsilon_p r
\end{align*}

In this model, the agent now has three decision variables, each with a different information set.
\begin{itemize}
    \item At $q$, the agent updates their estimate of the state of the latent income process based on the new information at $c$.
    \item At $d$, the agent makes a consumption decision, expending scarce resources in order to experience reward at $u$.
    \item At $r$, the agent makes a global decision about how much information they will retain from period to period, as $r$ limits the drift of memory caused by $\epsilon_p$. The agent can choose to have more accurate memory, but at the cost of utility at $v$.
\end{itemize}

This is well-formulated as an SCDP. It is not possible to model this problem as a POMDP. Unlike in the POMDP, we have endogenized the agent's memory and belief formation into the model, making these subject to strategies and cognitive costs.

\begin{figure}
\centering
\begin{tikzpicture}[
    state/.style={circle, draw, minimum size=0.9cm, font=\small},
    decision/.style={rectangle, draw, minimum size=0.9cm, font=\small},
    utility/.style={diamond, draw, minimum size=0.7cm, font=\small, inner sep=1pt},
    bridge/.style={circle, draw, very thick, minimum size=0.9cm, font=\small},
    noise/.style={circle, draw, double, minimum size=0.9cm, font=\small},
    >=stealth
]

\node[state] (a) at (0,0) {$a$};
\node[state] (a') at (6,0) {$a'$};
\node[noise] (eps_a) at (4,1) {$\epsilon_a$};
\node[state] (b) at (0, -2) {$b$};
\node[state] (c) at (1.5, -2) {$c$};
\node[decision] (d) at (3.5, -3) {$d$};
\node[state] (b') at (6, -2) {$b'$};
\node[noise] (eps_b) at (4,-1) {$\epsilon_b$};
\node[utility] (u) at (6, -3) {$u$};

\node[state] (p) at (0, -5) {$p$};
\node[decision] (q) at (2.5, -5) {$q$};
\node[decision] (r) at (4.5, -4) {$r$};
\node[utility] (v) at (6, -4) {$v$};
\node[noise] (eps_p) at (4,-6) {$\epsilon_p$};
\node[state] (p') at (6, -5) {$p'$};

\draw[->] (a) -> (a');
\draw[->] (eps_a) -> (a');
\draw[->] (b) -> (c);
\draw[->] (a) -> (c);
\draw[->] (c) -> (d);
\draw[->] (d) -> (b');
\draw[->] (c) -> (b');
\draw[->] (eps_b) -> (b');
\draw[->] (d) -> (u);

\draw[->] (p) -> (q);
\draw[->] (c) -> (q);
\draw[->] (q) -> (d);
\draw[->] (r) -> (v);
\draw[->] (r) -> (p');
\draw[->] (eps_p) -> (p');
\draw[->] (q) -> (p');

\end{tikzpicture}
\caption{An SCDP in which the agent consumes at $d$, updates their beliefs about the world at $q$, and chooses how much to remember at $r$. The agent experiences joy with consumption $u$ and consternation at remembering carefully $v$. Their resources vary over time due to transitory shocks $\epsilon_b$ as well as a latent stochastic income process which is not directly observed.}
\label{fig:example-scdp-costly-memory}
\end{figure}
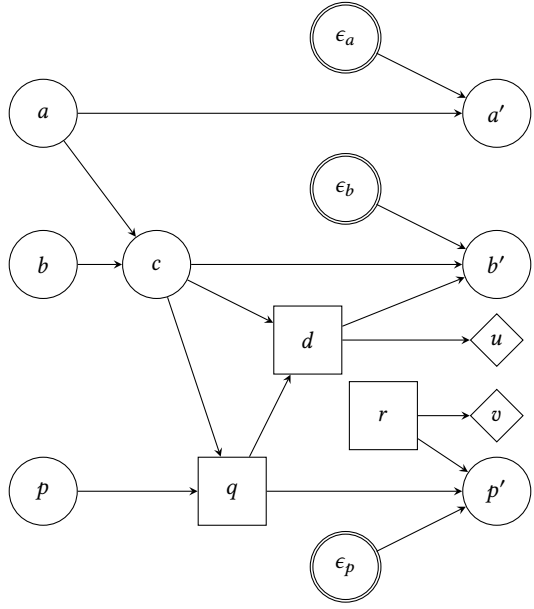

\section{Variable discount factors in SCDPs}
\label{sec:scdp-variable-discount}

In Definition \ref{def:scdp}, we discussed that the discount factor for an SCDP is a choice of variable $\beta \in \mathbf{V}$.
In most presentations of dynamic programming problems, the discount factor $\beta$ (or $\gamma$) is a constant; this is reflected in Definition \ref{def:dsop}.
However, in economics there is a widespread use of dynamic, state dependent, and non-exponential discounting, which we discussed in Section \ref{sec:lit-review-variable-discounting}.

Variable discounting poses no fundamental obstacle for the SCDM representation. Because the discount factor is a selected state variable $\beta \in \mathbf{V}$, it may have its own structural equation governing its evolution. (It may also be assigned a constant.) For example, in the quasi-hyperbolic case, $\beta$ takes different values depending on the temporal distance of the decision. In the state-dependent case of \citet{stachurski2021dynamic}, $\beta = \beta(x)$ for some function of the state. The graphical structure of the SCDM then makes explicit how discounting interacts with other state variables, potentially revealing new decomposition opportunities or, conversely, identifying when variable discounting breaks orthomodularity.

\subsection{Example: Stochastic discount factors}
\label{sec:stochastic-discount}

We now present an example with stochastic discount factors following \citet{krusell1998income}. In this formulation, the discount factor $\beta$ is a separate state variable that follows its own Markov process, independent of the income shock, allowing patience to vary across agents and over time.

Consider a consumption-saving problem where the discount factor $\beta \in \{\underline{\beta}, \bar{\beta}\}$ transitions according to a Markov chain with transition matrix $\Pi_\beta$. The persistence of $\beta$ is calibrated so that the average duration in each state corresponds to an agent's lifetime. The agent maximizes:
$$\max_{\{c_t\}} E \left[ \sum_{t=0}^{\infty} \left(\prod_{s=0}^{t-1} \beta_s\right) u(c_t) \right]$$

subject to the following transition equations and constraints:

\begin{align}
    y' &= \rho y + \epsilon_{y} ;\ \epsilon_{y} \sim \mathcal{N}(0,\sigma^2)\\
    \beta' &= g_\beta(\beta, \epsilon_{\beta}) ;\ \epsilon_{\beta} \sim F_\beta\\
    0 \leq c(a,y, \beta) &\leq a + y\\
    u &= u(c)\\
    a' &= a + y - c
\end{align}

where $a$ is the level of assets, $y$ is income, $\beta$ is the discount factor, $u$ is utility from consumption, and $g_\beta$ maps the current discount factor and shock to the next period's value according to the Markov transition $\Pi_\beta$. The income shock follows an AR(1) process, while $\beta' | \beta$ evolves independently. We assume $R = 1$ for simplicity.

The Bellman equation for this problem is:
\begin{equation}
v(a, y, \beta) = \max_{c \in [0, a+y]} \left\{ u(c) + \beta \cdot E_{y', \beta'}[v(a + y - c, y', \beta')] \right\}
\label{eq:stochastic-discount-bellman}
\end{equation}

The expectation is over both $y'$ and $\beta'$, which evolve independently according to their respective stochastic processes. Figure~\ref{fig:stochastic-discount} shows the influence diagram. The agent observes the state $(a, y, \beta)$ and chooses consumption $c$. The law of motion $a' = a + y - c$ determines next period's assets from current resources minus consumption. The bridge variables $(a', y', \beta')$ appear in overlapping boxes to indicate they belong to both the current period (as outputs of the transition) and the next period (as initial states).

\begin{figure}[h]
\centering
\begin{tikzpicture}[
    state/.style={circle, draw, minimum size=0.9cm, font=\small},
    decision/.style={rectangle, draw, minimum size=0.9cm, font=\small},
    utility/.style={diamond, draw, minimum size=0.7cm, font=\small, inner sep=1pt},
    noise/.style={circle, draw, double, minimum size=0.9cm, font=\small},
    discount/.style={regular polygon, regular polygon sides=5, draw, minimum size=1.1cm, font=\small, inner sep=1pt},
    >=stealth
]

\node[state] (a) at (0, 0) {$a$};
\node[state] (a') at (6, 0) {$a'$};
\node[state] (a'') at (12, 0) {$a''$};

\node[state] (y) at (0, -2) {$y$};
\node[state] (y') at (6, -2) {$y'$};
\node[state] (y'') at (12, -2) {$y''$};

\node[noise] (eps_y) at (4.5, -3.2) {$\epsilon_{y}$};
\node[noise] (eps_y') at (10.5, -3.2) {$\epsilon_{y'}$};

\node[decision] (c) at (2, -4) {$c$};
\node[utility] (u) at (3.5, -4) {$u$};
\node[decision] (c') at (8, -4) {$c'$};
\node[utility] (u') at (9.5, -4) {$u'$};

\node[discount] (beta) at (0, -6) {$\beta$};
\node[discount] (beta') at (6, -6) {$\beta'$};
\node[discount] (beta'') at (12, -6) {$\beta''$};

\node[noise] (eps_beta) at (4.5, -7.2) {$\epsilon_{\beta}$};
\node[noise] (eps_beta') at (10.5, -7.2) {$\epsilon_{\beta'}$};

\draw[->] (a) -- (c);
\draw[->] (y) -- (c);
\draw[->] (beta) -- (c);

\draw[->] (c) -- (u);

\draw[->] (a) -- (a');
\draw[->] (c) -- (a');
\draw[->] (y) -- (a');
\draw[->] (y) -- (y');
\draw[->] (eps_y) -- (y');
\draw[->] (beta) -- (beta');
\draw[->] (eps_beta) -- (beta');

\draw[->] (a') -- (c');
\draw[->] (y') -- (c');
\draw[->] (beta') -- (c');

\draw[->] (c') -- (u');

\draw[->] (a') -- (a'');
\draw[->] (c') -- (a'');
\draw[->] (y') -- (a'');
\draw[->] (y') -- (y'');
\draw[->] (eps_y') -- (y'');
\draw[->] (beta') -- (beta'');
\draw[->] (eps_beta') -- (beta'');

\begin{scope}[on background layer]
    \node[draw=gray, dashed, rounded corners, fit=(a)(y)(beta)(c)(u)(a')(y')(beta')(eps_y)(eps_beta), inner sep=0.4cm] (box1) {};
    \node[draw=gray, dashed, rounded corners, fit=(a')(y')(beta')(c')(u')(a'')(y'')(beta'')(eps_y')(eps_beta'), inner sep=0.4cm] (box2) {};
\end{scope}

\node[below=0.3cm of box1.south, font=\footnotesize\itshape] {Period $t$};
\node[below=0.3cm of box2.south, font=\footnotesize\itshape] {Period $t+1$};

\end{tikzpicture}
\caption{Influence diagram for the consumption-saving problem with stochastic discount factors. The agent observes $(a, y, \beta)$ and chooses consumption $c$. The bridge variables $(a', y', \beta')$ connect to the next period; the boxes overlap on these nodes to indicate they belong to both components. Income and discount factor each evolve via Markov processes driven by exogenous shocks. Pentagons denote discount factors; double circles denote exogenous noise variables.}
\label{fig:stochastic-discount}
\end{figure}
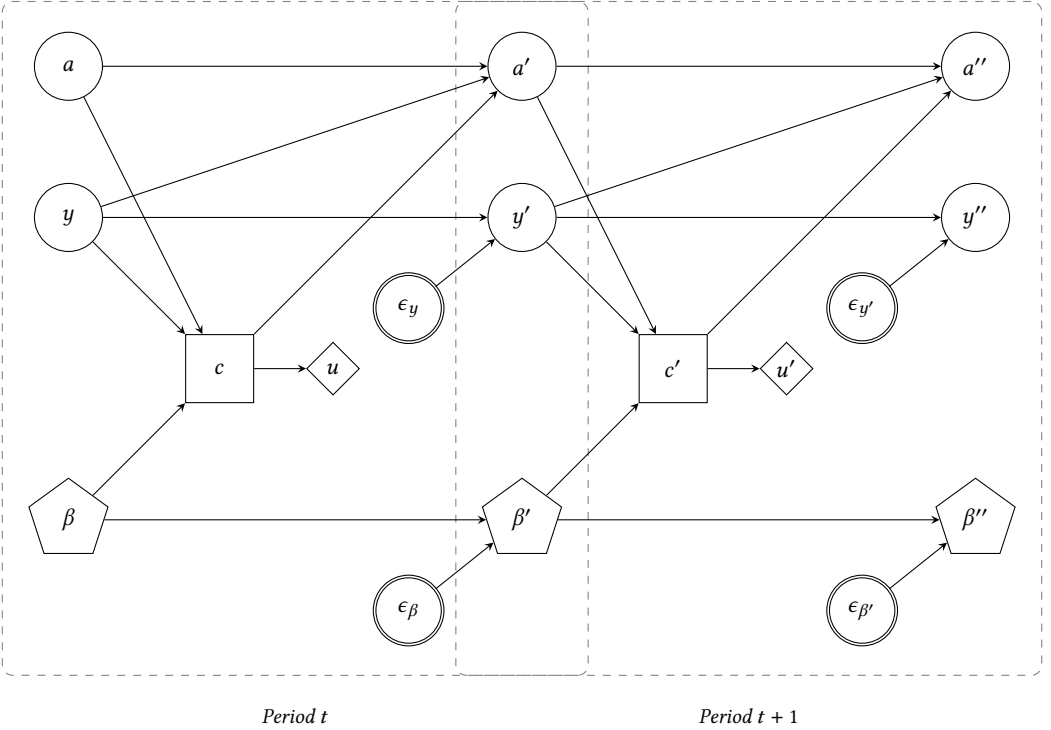

This model decomposes in the standard way. The state $a'$ depends on $(a, y, c)$ via the law of motion $a' = a + y - c$, while $\beta'$ depends only on $\beta$ through its stochastic transition, not on endogenous choices. Since both $y$ and $\beta$ evolve exogenously, they do not create additional channels between periods, and the orthomodularity condition remains satisfied. When $\beta = \bar{\beta}$ for all $t$, this reduces to the standard model. The stochastic formulation preserves the recursive structure that enables dynamic programming.

\section{Conclusion}

In this paper, we have approached the problem of modeling agents in computing systems with several desiderata in mind. These include:
\begin{itemize}
    \item Graphical causal modeling precisely captures the structural relationships between variables.
    \item The agents are goal directed, but they may be imperfectly informed, have insufficient recall of their past, and can act suboptimally.
    \item The models are (de)composability, such that the analyst can capture efficiency gains due to the strategic independence of control variables.
    \item The models can represent latent, unobserved state and multiple decision variables with different information sets within a single model or dynamic period.
    \item Limitations to cognitive resources and value discounting are endogenous to the model.
\end{itemize}

We introduce new modeling tools, SCDMs and SCDPs, which satisfy all of the above criteria.
These tools build on recent work in causal game theory, and extend it to the dynamic setting.
SCDMs are distinguished from prior constructs like SCIMs mainly because of how they allow for root variables to be ungoverned by probability distributions or structural equations.
These root nodes are what enable SCDMs to be composed and also converted into SCDPs.
We show that SCDPs are more expressive than POMDPs, particularly with respect to their ability to model resource rationality and variable discount factors.

In future work, we will extend SCDMs and SCDPs into multi-agent systems.
The groundwork has already been laid for this with causal game theory and Structural Causal Games.
We will also develop efficient algorithms for solving these models and fitting them to data that take advantage of the causal structure.
We are developing a scientific software library, \texttt{scikit-agent}, which implements this modeling framework and will later be a repository for implementations of these algorithms.

\begin{acks}
SB's contributions have been supported by NSF awards 2131532, 2131533 and 2105301, as well as the Future of Life Foundation, the Microsoft Corporation, and the Econ-ARK project. AL's contributions have been supported by NSF award 2131533.
Any opinions, findings, and conclusions or recommendations expressed in this paper are those of the authors and do not necessarily reflect the views of the sponsors.
\end{acks}

\bibliographystyle{ACM-Reference-Format}
\bibliography{references}

\appendix

\section{Reference}

\subsection{Structural Causal Games}

A Bayesian network is a graphical model of a joint distribution over random variables. 

\begin{definition}[Causal Bayesian Network (CBN) \citep{pearl2009causality}]\label{def.bayesian.network}
A causal Bayesian network (CBN) over set of random variables $\mathbf{V}$, parameters $\theta$, and joint probability distribution $\Pr(\mathbf{V}; \theta)$ is
a structure $\mathcal{M} = (\mathbf{V}, \mathcal{E}, \theta)$
such that:
\begin{itemize}
    \item $\mathcal{G} = (\mathbf{V}, \mathcal{E})$ is a directed acyclic graph with vertices $\mathbf{V}$ and edges $\mathcal{E}$, and
    \item $\Pr(\mathbf{v}; \theta) = \prod_{V \in \mathbf{V}} Pr(v | \mathbf{pa}_V ; \theta_V)$, where $\mathbf{pa}_V$ are the parents of $V$ on $\mathcal{G}$.
\end{itemize}
\end{definition}

The marginal probability distribution of each variable is conditional only on its parents. Together, these marginal conditional distributions constitute the original joint distribution $\Pr_\theta(V)$.
While this distribution can be represented by many different (Markov-equivalent) Bayesian network structures, the edge structure $\mathcal{E}$ is considered \textit{causal} with respect to interventions on the variable values \cite{pearl2009causality}.

A \emph{structural causal model} (SCM) is a Bayesian Network that has all of its stochasticity in its root nodes.
The exogenous variables -- those which have no parents -- are random variables.
The values of the endogenous variables -- which do have parent nodes -- are governed by structural equations or, equivalently, deterministic conditional probability distributions.

\begin{definition}[Structural Causal Model (SCM) \citep{pearl2009causality}]\label{def:scm}
A (Markovian) structural causal model is a CBN 
$\mathcal{M} = (\mathbf{W}, E, \boldsymbol{\theta})$
where $\mathbf{W} = (\mathbf{V} \cup \mathbf{Z})$,
with exogenous variables $\mathbf{Z}$ and endogenous variables $\mathbf{V}$ such that for all $Z \in \mathbf{Z}$, $\mathbf{pa}_Z = \emptyset$ and for all $V \in \mathbf{V}$, $\mathbf{pa}_Z \neq \emptyset$.
The parameters $\boldsymbol{\theta}$ assign deterministic distributions $\text{Pr}(V | \mathbf{pa}_V ; \theta_V)$ to each endogenous variable and a stochastic distribution $\text{Pr}(\mathbf{Z}; \boldsymbol{\theta}) = \prod_{Z \in \mathbf{Z}} \text{Pr}(Z; \theta_Z)$ to the
exogenous variables.
\end{definition}

A \emph{structural causal game} (SCG) is an SCM that reserves some of its variables as decision and utility variables, assigned to particular agents.

\begin{definition}[Structural Causal Game \citep{koller2003multi, hammond2023reasoning}]
\label{def:scg}
A structural causal game (SCG) is a structure $\mathcal{G} = (N, \mathbf{W}, E, \boldsymbol{\theta})$ where
\begin{itemize}
    \item $N = \{1, \ldots, n\}$ is a set of agents
    \item $(\mathbf{W}, E, \boldsymbol{\theta})$ is an SCM with endogenous variables $\mathbf{V} \subset W$.
    \item $\mathbf{V}$ is partitioned into nature variables $\mathbf{X}$, decision variables $\mathbf{D} = \cup_{i \in N} \mathbf{D}_i$, and utility variables $\mathbf{U} = \cup_{i \in N} \mathbf{U}_i$
\end{itemize}
\end{definition}

A single-player variation of an SCG is the SCIM work of \citet{hammond2023reasoning}:

\begin{definition}[Structural Causal Influence Model (SCIM) \cite{hammond2023reasoning} ]
\label{def:scim}
A \emph{structural causal influence model} $(\mathbf{V}, \mathbf{E}, \mathcal{E}, (\mathbf{X}, \mathbf{D}, \mathbf{U}), \mathbf{Pr}, \boldsymbol{\theta})$ consists of:
\begin{itemize}
\item A set of endogenous variables $\mathbf{V}$
\item A set of exogenous variables $\mathbf{E} = \{E_V\}_{V \in \mathbf{V}}$.
\item A graph $\mathcal{G} = (\mathbf{E} \cup \mathbf{V}, \mathcal{E})$ is a DAG over $\mathbf{E}$ and $\mathbf{V}$ where $\mathbf{Pa}_V \cap \mathbf{E} = \{E_V\}$
\item $\mathbf{V}$ is partitioned into: \begin{itemize}
\item $\mathbf{X}$, state variables
\item $\mathbf{D}$, decision variables\footnote{The influence diagram literature rarely intersects with the control theory literature. Decision variables and control variables are roughly synonymous.}
\item $\mathbf{U}$, utility variables
\end{itemize}
\item  Distribution $\mathbf{Pr}$ and parameters $\boldsymbol{\theta}$ assign deterministic distributions $Pr(x | \mathbf{pa}_x, \theta_X)$ to each state variable \footnote{The use of probability distributions for the deterministic variable assignments are a consequence of this construct's being derived from work on Bayesian networks. We will address this directly in a moment.},  $Pr(u | \mathbf{pa}_u, \theta_U)$ to each utility variable, and a stochastic distribution $\mathbf{Pr}(\mathbf{E}, \boldsymbol{\theta}) = \prod_{E \in \mathbf{E}} Pr(E; \theta_E)$ to the exogenous variables.
\end{itemize}
\label{def:scim-apx}
\end{definition}

The SCDM definition used in this paper is indebted to this prior work and is an extension of these constructs.

\subsection{Strategic reliance}
\label{sec:strategic reliance}

For multi-agent causal influence models (of which the  SCIMs of Definition \ref{def:scim} are a subclass), \citet{koller2003multi} identify graphical criteria for the strategic relevance of decision variables on each other, which allow for efficient algorithms for finding the Nash equilibria. SCIMs are a simpler problem, and the same logic applies.

\begin{definition}[Strategic reliance \cite{koller2003multi}]
\label{def:strategic-reliance}
A decision variable $D$ \emph{strategically relies} on a
decision variable $D'$ in an SCIM if there are two policy
profiles $\pi$ and $\pi'$ such that $\pi$ and $\pi'$ differ only at $D'$, but some decision rule for $D$ is optimal for $\pi$ and not for $\pi'$.
\end{definition}

For an SCIM, let $U_D = \mathbf{U} \cap \text{Desc}(D)$ be those utility variables that are descendants of decision variable $D$.

\begin{definition}[S-reachability \cite{koller2003multi}]
A node $D'$ is \emph{s-reachable} from a node $D$ in an
SCIM if there is some utility node $U \in U_D$ such that if a
new parent $\hat{D'}$ were added to $D'$, there would be an active
path from $\hat{D'}$ to $Pa(U)$ given $Pa(D) \cup {D}$.
\label{def:s-reachability}
\end{definition}

The definition of \emph{active path} comes from earlier work on causal graphical models.

\begin{definition}[d-separation; active path \cite{verma1988influence}]
\label{def:d-separation}
On a directed graph, two sets of nodes $X$ and $Y$ are d-separated given a third set $Z$ if and only if there is no active bi-directed path from a node in $X$ to a node in $Y$.

A path between two nodes is an \emph{active path} given a set of nodes $Z$ if every node with converging arrows (a collider) either is or has a descendant in $Z$ and every other node along the path is not in $Z$.
\end{definition}

\citet{koller2003multi} then prove the connection between S-reachability and strategic reliance.

\begin{theorem}[Soundness \cite{koller2003multi}]
If $D$ and $D'$ are two decision nodes in a SCIM and $D'$ is not s-reachable from $D$ in the
MAID, then $D$ does not rely on $D'$.
\label{thm:km-soundness}
\end{theorem}

\citet{koller2003multi} outline a general algorithm for identifying the reliance structure of decision nodes in a multi-agent influence diagram and efficiently solving for their optimal and equilibrium strategy profiles.

\end{document}